\documentclass[conference,compsoc]{IEEEtran}

\ifCLASSOPTIONcompsoc
  \usepackage[nocompress]{cite}
\else
  \usepackage{cite}
\fi
\ifCLASSINFOpdf
\else
\fi
\usepackage{tikz}
\usepackage{amssymb}
\usepackage{amsmath}
\usepackage{mathtools}
\usepackage{cite}
\usepackage{graphicx}
\usepackage{subfigure}
\usepackage{algorithm}
\usepackage{algorithmic}
\usepackage{amsthm}
\usepackage{thmtools, thm-restate}
\usepackage{multirow}
\usepackage{array}
\usepackage{caption}
\usepackage{soul}
\usepackage{pgfplots}
\usepgfplotslibrary{statistics}
\usepgfplotslibrary{groupplots}
\usepackage{titlesec}
\usepackage{nicematrix}
\usepackage{hyperref}

\usepackage{xcolor}
\usepackage{mdframed}
\definecolor{mygray}{RGB}{211,211,211} 

\newcommand{\attack}{\textsc{HiDRA}}
\newcommand{\defence}{robust aggregator}

\newcommand{\DEFENCES}{Robust Aggregators}
\newcommand{\defences}{robust aggregators}
\newcommand{\Defences}{Robust aggregators}

\newcommand{\bias}{bias}

\newcommand{\threshold}{\xi}

\newtheorem{theorem}{Theorem}

\theoremstyle{definition}

\newtheoremstyle{proofsketch}
  {\topsep} 
  {\topsep} 
  {\normalfont} 
  {} 
  {\itshape} 
  {.} 
  {.5em} 
  {} 

\theoremstyle{proofsketch}
\newtheorem*{proofsketch}{Proof Sketch}

\titleformat{\paragraph}[runin]{\bfseries}{}{0pt}{}[]
\titlespacing{\paragraph}{0pt}{\baselineskip}{1em}
\begin{document}
%
\title{\Large \bf Attacking Byzantine Robust Aggregation in High Dimensions}

\author{\IEEEauthorblockN{Sarthak Choudhary\textsuperscript{*},
Aashish Kolluri\textsuperscript{*} and
Prateek Saxena}
\IEEEauthorblockA{School of Computing,
National University of Singapore\\
csarthak76@gmail.com 
\{aashish7, prateeks\}@comp.nus.edu.sg
}}



%


\maketitle
\begingroup\renewcommand\thefootnote{*}
\footnotetext{These authors contributed equally to this work.}
\endgroup
\begin{abstract}
Training modern neural networks or models typically requires averaging over a sample of high-dimensional vectors. Poisoning attacks can skew or bias the average vectors used to train the model, forcing the model to learn specific patterns or avoid learning anything useful. Byzantine robust aggregation is a principled algorithmic defense against such biasing. Robust aggregators can bound the maximum bias in computing centrality statistics, such as mean, even when some fraction of inputs are arbitrarily corrupted. Designing such aggregators is challenging when dealing with high dimensions. However, the first polynomial-time algorithms with strong theoretical bounds on the bias have recently been proposed. Their bounds are {\em independent} of the number of dimensions, promising a conceptual limit on the power of poisoning attacks in their ongoing arms race against defenses.

In this paper, we show a new attack called \attack~on practical realization of strong defenses which subverts their claim of dimension-independent bias. \attack~highlights a novel computational bottleneck that has not been a concern of prior information-theoretic analysis. Our experimental evaluation shows that our attacks almost completely destroy the model performance, whereas existing attacks with the same goal fail to have much effect. Our findings leave the arms race between poisoning attacks and provable defenses wide open.

\end{abstract}


%
\IEEEpeerreviewmaketitle

\section{Introduction}
\label{sec:intro}

Machine learning training algorithms often have to compute an average over a set of vectors. The standard stochastic gradient descent (SGD) algorithm for neural network training, as an example, computes the average of gradient vectors derived from data samples in each step. Similarly, in federated learning, local models are trained at individual worker machines and then sent to a central service that averages over these vectors to get the global model. 

Several security problems arise if some of these vectors can be maliciously crafted. For instance, poisoning attacks corrupt training data samples so that gradients computed from them during SGD skew or bias the learned model. The {\em bias} is how far the average computed from the partially corrupted vectors can be from the benign (uncorrupted) value.
Such biasing attacks on ML training can severely deteriorate the training accuracy of the model~\cite{xie2019dba,fang2020local,xie2020fall}, exacerbate privacy concerns~\cite{tramer2022truth}, and create unfairness that did not exist in the dataset~\cite{solans2020poisoning,zhang2022pipattack}.
Local models sent by malicious worker machines in federated learning can directly be corrupted before being sent to the central service.
%


A natural security question arises about the robustness of such averaging. It is easy to see that even a single corrupted input vector can result in arbitrarily biasing the average when taking a simple arithmetic mean. 
%
%
It is desirable to design algorithms that compute an {\em aggregate} statistic over given vectors, similar to the mean, which cannot be biased much even by a strong adversary. 
The {\em Byzantine Robust Aggregation} problem captures this goal: Imagine a fraction $\epsilon$ of vectors can be arbitrarily corrupted, then can we ensure the bias in the computed aggregate is small~\cite{huber1992robust, tukey1960survey}? Robust aggregation algorithms bound how much the attacker can skew the average model at each step in training, thereby offering a principled limit on the effect of attacks. 


Designing \defences~is an algorithmic challenge that has inspired many attempts. These aggregators, as shown in Table~\ref{tab:robust_gaurantees}, most often provide a 
 theoretical upper bound
\footnote{The more precise asymptotical bounds are in Table~\ref{tab:robust_gaurantees}. We will often drop the $\Sigma$ term, as it is a constant for a given distribution.}
 on the maximum bias of $O(\sqrt{\epsilon d})$~\cite{yin2018byzantine, pillutla2022robust, chen2020distributed, blanchard2017machine}.
%
%
%
The bias is the $L_2$ norm of the maximum adversarial error induced, $\epsilon$ is the fraction of vectors controlled by the adversary, and $d$ is the number of dimensions of the vectors. The issue in practice is that the number of dimensions $d$ is usually the parameter size of the ML model, which can be in millions or billions. The bias bound being dependent on $d$ makes it vacuously large for modern ML systems. Hence, we refer to these algorithms as weakly bounded (or {\em weak}). Concrete poisoning attacks that have a severe impact on weak algorithms have been shown recently~\cite{fang2020local,xie2020fall, shejwalkar2021manipulating}. 

On the theoretical side, robust aggregation or robust mean estimation has been a long-standing challenge in high-dimensional statistics.
The primary goal is to give {\em stronger bounds of $O(\sqrt{\epsilon})$ on the bias}, without the $\sqrt{d}$ factor, which is the statistically optimal ~\cite{zhu2023byzantine, diakonikolas2019recent}. The classical Tukey median~\cite{tukey1960survey} achieves the goal but best known algorithms for computing it have running time exponential in $d$~\cite{amaldi1995complexity}. A new line of algorithms gives the first polynomial time solutions with strong bounds on bias~\cite{diakonikolas2017being} and their first practical realization strategy is shown recently~\cite{zhu2023byzantine}. Table~\ref{tab:robust_gaurantees} summarizes the guarantees of strong aggregators.
It has been experimentally verified that they completely thwart existing poisoning attacks which aim to reduce the trained model performance, thus establishing the state-of-the-art~\cite{zhu2023byzantine}. Importantly, these are the only polynomial-time algorithms known to have strong bounds on the bias independent of $d$---a crucial security property when working with high-dimensional vectors as in modern ML training.

\paragraph{Our work.} Since none of the existing attacks defeat strong \defences, even when working in low dimensions, it is natural to ask: Are there {\em any} attacks that create bias matching the theoretical upper bound of $O(\sqrt{\epsilon})$ in them? 
In this paper, we present the first effective attack against strong \defences, called \attack\footnote{\attack~is short for \underline{\textbf{Hi}}gh \underline{\textbf{D}}imensional attack on \underline{\textbf{R}}obust \underline{\textbf{A}}ggregators.}. \attack~induces bias matching their analytical upper bounds in low dimensional settings. Our attack, thus, shows that the prior theoretical bounds for strong defenses are tight.

\begin{table}[t]
    \centering
    \caption{Comparison of the worst-case bias between different robust aggregation algorithms. $||\Sigma||_2$ is maximum variance of the uncorrupted sample. $\tilde{O}( \cdot)$ ignores the constant and logarithmic factors in the computation complexity. The number of vectors is $n$ and dimensions $d$.} 
    \label{tab:robust_gaurantees}
    \renewcommand{\arraystretch}{1.5}
    \resizebox{\columnwidth}{!}{
    \begin{tabular}{|c|c|c|}
        \hline
        \textbf{Algorithm} & \textbf{Max. Bias} & \textbf{Comp. Complexity} 
        \\
        \hline
        Weak~\defences& &
        \\
        \cline{1-1} 
        \label{row:suboptimal_start}Median ~\cite{yin2018byzantine} & $\tilde{O}(\sqrt{\epsilon d}) \cdot ||\Sigma||_2^{\frac{1}{2}}$ & $\tilde{O}(nd)$
        \\
       Trimmed Mean ~\cite{yin2018byzantine}& $\tilde{O}(\sqrt{\epsilon d}) \cdot ||\Sigma||_2^{\frac{1}{2}}$ & $\tilde{O}(nd)$
        \\
        \label{row:suboptimal_end} Geometric Median ~\cite{pillutla2022robust, blanchard2017machine, chen2020distributed} & $\tilde{O}(\sqrt{\epsilon d}) \cdot ||\Sigma||_2^{\frac{1}{2}}$ & $\tilde{O}(nd)$
        \\
        \hline
        Strong~\defences& &
        \\
        \cline{1-1} 
    \label{row:optimal_start}Filtering~\cite{diakonikolas2017being}& $\tilde{O}(\sqrt{\epsilon}) \cdot ||\Sigma||_2^{\frac{1}{2}}$ & $\tilde{O}(\epsilon n \cdot d^3)$ \\
        No-Regret~\cite{hopkins2020robust} & $\tilde{O}(\sqrt{\epsilon}) \cdot ||\Sigma||_2^{\frac{1}{2}}$ & $\tilde{O}((n + d^3)\cdot d)$ 
        \\
        \label{row:optimal_end}SoS ~\cite{kothari2017outlierrobust} & $\tilde{O}(\sqrt{\epsilon}) \cdot ||\Sigma||_2^{\frac{1}{2}}$ & poly$(n, d)$ 
        \\
        Tukey Median~\cite{tukey1960survey} & $\tilde{O}(\sqrt{\epsilon}) \cdot ||\Sigma||_2^{\frac{1}{2}}$ & NP-Hard in $d$ 
        \\
        \hline
    \end{tabular}
    }
\end{table}

More importantly, we observe that the bounds given by prior theoretical analysis make {\em idealized computational assumptions} which hold primarily when the number of dimensions is small. As the number of dimensions increases, existing \defences~run into a fundamental computation bottleneck.
Practical realizations of these defenses, therefore, when working with high dimensions, have to break down the given vectors into multiple low-dimensional chunks to solve for. 
Our \attack~attack induces a near optimal\footnote{The known upper bound is $\tilde{O}(\sqrt{\epsilon d})$, see Table~\ref{tab:robust_gaurantees}.} bias per chunk, resulting in a total bias of $\Omega(\sqrt{\epsilon d})$ in the high dimensional setting. {\em This is in sharp contrast to the promised dimension independent $\tilde{O}(\sqrt{\epsilon})$ bias}.
A factor of $\sqrt{d}$ translates to several orders of magnitude worse bias, even for moderately sized neural networks that have a million parameters.
We analytically derive the lower bound on the bias achieved by our \attack~attack on the above chunking procedure of $\Omega(\sqrt{\epsilon d})$ and experimentally confirm that corruption of input vectors using \attack~hits this lower bound. \attack~is thus, again, near-optimal for high $d$. 
%


%
%

\paragraph{Experimental results.}
We employ \attack~towards creating indiscriminate or untargeted poisoning attacks as a concrete application~\cite{aisec2017}. Untargeted poisoning attacks aim to prevent ML models from learning useful information and have been difficult to achieve with prior attacks when strong~\defences~are used.
However, our attack consistently results in a drastic drop in the accuracy of trained models even when using strong \defences. For example, in several instances \textbf{the original model accuracy of over $\mathbf{80\%}$ drops to below $\mathbf{10\%}$ with $\mathbf{\epsilon=0.2}$ fraction of vectors corrupted} using \attack. Under the same setup, prior attacks induce {\em below $5\%$} drop in performance.

\paragraph{Computational Bottleneck is fundamental.} 
The computational bottleneck targeted by \attack~is fundamental to the problem of robust aggregation in general, not specific to a single algorithm.
The key idea in all strong \defences~is to filter out vectors far from the mean in the {\em direction of maximum variance} of the given vectors. For this, strong aggregators compute the maximum variance direction. It turns out that this is no accident. We show a quasilinear time reduction from the problem of computing the maximum variance direction approximately to that of strong robust aggregation, for a non-trivial class of inputs (see Section~\ref{sec:tradeoff}). In other words, solving strong aggregation efficiently would imply finding the approximate maximum variance direction efficiently for this class of inputs.


To our knowledge, the best-known algorithms for computing the maximum variance direction of given vectors have $O(d^3)$~\footnote{Algorithms for exact maximum variance run in $O(d^3)$, while the approximate algorithms with desired error rates operate in $\tilde{O}(n^2d)$. We refer readers to section~\ref{sec:tradeoff} for more details.} time complexity~\cite{cuppen1980divide,arbenz2012lecture}
%
Their highly optimized implementations\footnote{https://numpy.org/doc/stable/reference/generated/numpy.linalg.eigh.html} remain $O(d^3)$, despite being widely used for decades.
For modern ML models, $d$ can be $2^{20}$ to $2^{40}$, so $O(d^3)$ is prohibitively costly.
The computational bottleneck in strong aggregators, therefore, does not appear easy to work around in full generality.

\paragraph{Our contribution.} 
%
We present the first effective attacks against provable byzantine robust aggregation defenses that have the strongest bias bounds. Our attacks are {\em near optimal} and experimentally surpass the efficacy of existing untargeted poisoning attacks for these defenses. They create a severe loss in performance of trained models. The vulnerability is {\em computational}, highlighting the gaps between idealized information-theoretic analysis of these defenses and their practical realization. As a result, designing practical and provable defenses that limit the bias in poisoning attacks remains an open question for the future.


\section{Background \& Problem Setup}
\label{sec:backandprob}
The central question we are interested in is that of computing a ``mean-like" aggregate statistic of $n$ vectors when a small unknown fraction $\epsilon$ of them has been arbitrarily corrupted. The defender wants the computed aggregate to be {\em robust}, i.e., not too far from the mean of uncorrupted vectors. The problem arises in many applications, though we are motivated by its use in machine learning tasks. 


\subsection{The Problem: Byzantine Robust Aggregation}

There are $n$ vectors $X= \{ x_1, ..., x_n\}$, for example $n$ gradients of a neural network, where each vector $x_i \in \mathbb{R}^d$. A {\em Byzantine} adversary who has access to $X$ can replace an $\epsilon < \frac{1}{2}$ fraction of vectors in $X$ and with arbitrarily chosen vectors, resulting in a set $Y = \{y_1, ..., y_n\}$ called the $\epsilon$-corrupted set. The defender is given $Y$ and some auxiliary information about $X$. The robust aggregation problem is to specify an aggregation function $f: R^{n\times d}\rightarrow \mathbb{R}^d$, such that the {\em bias}---the difference between $f(Y)$ and the arithmetic mean of $X$ for all possible choices of $Y$---is bounded by $\lambda$. Formally, there exists some $\lambda>0$ such that:
\begin{align}
\bias = \max_{Y}{\lVert f(Y) - \frac{1}{n}\sum_{i=1}^n x_i \rVert_2} \leq \lambda \label{eq:bias}
\end{align}


There are known information-theoretic limits for which sets we can solve robust aggregation. Solving the above problem with provable guarantees in full generality, without making any assumptions about the {\em uncorrupted} samples $X$, is not possible. One must minimally assume that an upper bound on the variance of uncorrupted samples is known~\cite{diakonikolas2017being, lai2016agnostic, kothari2017outlierrobust}---without it, there exist sets where a small number of corruptions make it impossible to recover the original mean. We will therefore assume that a reasonably good estimate of maximum variance is known in advance. The maximum variance is the {\em largest eigenvalue} or the {\em spectral norm} of $\Sigma$, denoted by $||\Sigma||_2$, where $\Sigma \in \mathbb{R}^{d \times d}$ is the covariance matrix of $X$ defined below:

\begin{align*}
\Sigma = \frac{1}{n}\cdot\left[\sum_{i=0}^{n}{(x_i -\hat{\mu})^T \cdot (x_i - \hat{\mu})}\right]\text{;      }
\hat{\mu} = \frac{1}{n}\sum_{i=0}^{n}{x_i}
\end{align*}

The statistically optimal (best possible) bias $\lambda$ is a multiplicative factor, independent of $d$, times the maximum variance $||\Sigma||^{\frac{1}{2}}_2$. We refer to this multiplicative factor as the {\em bound} when clear from context. $||\Sigma||^{\frac{1}{2}}_2$ is fixed by the distribution $X$ is drawn from, but it may depend on $d$.

The robust aggregation problem is also called robust mean estimation in statistics~\cite{huber2004robust, tukey1960survey, diakonikolas2019robust, diakonikolas2019recent, lai2016agnostic, hopkins2020robust, kothari2017outlierrobust, kothariJacob, kothari2022polynomial}. The above assumption about knowing 
$||\Sigma||_2$ is standard in prior work on robust mean estimation~\cite{diakonikolas2017being}.

\subsection{Application: Untargeted Poisoning in SGD}
\label{sec:setup}

The standard training algorithm for neural networks is stochastic gradient descent (SGD)~\cite{bottou1991stochastic}.  We will use training with SGD as the baseline throughout the paper.

\paragraph{Aggregation in SGD.} SGD works by initializing the ML model parameters or weights $w \in \mathbb{R}^d$ with random values and iteratively updating it in sequential steps. In each step $t$, it selects a set of $n$ samples uniformly at random from the training dataset and computes the aggregate:  
$$ w^{t+1} = w^{t} - \frac{\eta}{n} \sum_{i=1}^n \nabla Q_i(w^t) $$

Here, the number of samples taken in each set is a constant $n$, the gradient of the $i^{th}$ data point in the set is $\nabla Q_i(w)$, and the learning rate $\eta$ is a constant. The gradients are vectors in $\mathbb{R}^d$. The choice of loss functions, neural network architecture, and so on are not relevant for the purpose of modeling the problem of robust aggregation.

The vanilla SGD procedure, as described above, is susceptible to adversarial corruptions of gradients. This can result in arbitrary bias since a simple arithmetic mean (average) of gradients is computed in each training step. A more robust aggregation technique than arithmetic mean is a natural mitigation for attacks that try to bias the aggregate.

\paragraph{Threat model.} 
We assume that the defender will use a strong \defence~$f$ instead of naive arithmetic mean during the SGD computation. These aggregators take as input the corrupted samples $Y$ and the auxiliary information $||\Sigma||_2$ about the benign samples. In practice, one can estimate $||\Sigma||_2$ experimentally based on a small set of clean unpoisoned samples. We assume the defender knows $||\Sigma||_2$ to use in $f$, and the attacker also knows $f$ and $||\Sigma||_2$.

The Byzantine attacker model admits a powerful attacker that can arbitrarily corrupt vectors being aggregated. Such an attack model abstracts away from how the attacker induces the corruption, for instance, whether through poisoning of datasets~\cite{aisec2017} or direct control over the vectors~\cite{shen2016auror}.
In typical centralized training, we are concerned with attackers who control training data samples and can poison them. Gradients computed from poisoned samples are thus influenced significantly by the attacker but not necessarily fully controlled. In contrast, certain federated learning setups offer more direct control over gradient vectors being aggregated.
In federated learning, a server collaborates to train a machine learning model with several clients or worker machines. Clients compute local updates (vectors) to the model on their local datasets and send them to the server, which aggregates them as in vanilla SGD. It is easy to see that a compromised or malicious worker has direct and complete control over the vector they send. Our Byzantine attacker model thus captures both the centralized and federated learning setups.

The adversary is {\em fully adaptive}, meaning that they corrupt with complete knowledge of the function $f$ and $||\Sigma||_2$ in every training round. The $\epsilon$-corrupted sample of $X$, denoted as $Y$, is given to the robust aggregator (defender).

\begin{figure*}[t]
    \centering
    \includegraphics[scale=0.75]{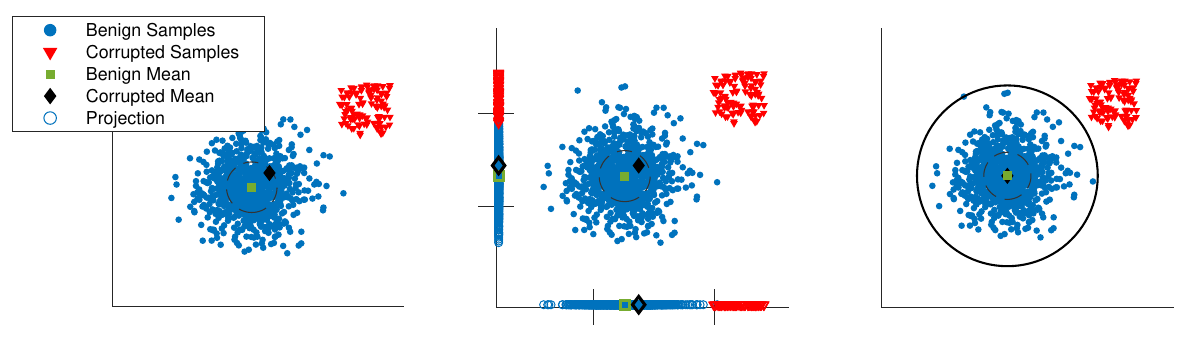}
    \caption{Left: Gaussian samples with 0.1 ($\epsilon$) fraction of arbitrarily corrupted data, highlighting the mean shift post-corruption where the dotted circle is the standard deviation ($\sigma$) boundary around the mean. Middle:  Trimmed mean by dimension, establishing dimension-wise thresholds to contain corrupted mean within an order of $d$. Right: Strong \defence~defenses, applying a single threshold based on variance to restrict corrupted mean to a constant distance.
}
\label{fig:weak-strong-defenses}
\end{figure*}

\paragraph{Untargeted model poisoning attacks.}
While there can be many attacker goals of biasing models during training, we focus solely on one type of attacks, namely {\em untargeted poisoning}~\cite{aisec2017}, to show practical value.
The attack corrupts the gradient vectors computed during training. The primary objective of untargeted poisoning is to bias the aggregated model at each step such that it exhibits a high error rate for training examples, misclassifying them to any class other than the correct one generically. Effectively, the ML model updates in each SGD step learn little useful information. The final trained model is unusable, leading to denial-of-service. 

Untargeted poisoning attacks are a subclass of poisoning attacks. There are $2$ other prominent categories: targeted~\cite{bhagoji2019analyzing} and backdoor~\cite{xie2019dba} poisoning. Targeted attacks force the trained ML model to misclassify a specific class of inputs. Backdoor attacks misclassify inputs with artificially planted trigger patterns.  We do not study these $2$ categories here, though we believe our key ideas can be adapted to them.

Prior untargeted attacks assume two types of attacker knowledge about the uncorrupted samples $X$, i.e., full-knowledge and partial knowledge. In the full knowledge setting, the adversary has direct access to $X$ and is able to arbitrarily manipulate an $\epsilon < 1/2$ fraction of inputs from $X$. In the partial-knowledge setting, the adversary does not have access to the complete $X$. In particular, for instance, it cannot exactly compute the mean of $X$ since it only knows its own uncorrupted vectors, not those of other clients, as is expected in typical federated learning setups. 
The rest of the capabilities in the partial-knowledge setup are the same as in the full-knowledge setting.
We experimentally evaluate both these settings.
Our formal analysis, though, is restricted to the full-knowledge case because it is difficult to generically say how much the subset known to the adversary statistically deviates from the full $X$ in different practical setups.

\section{\DEFENCES}
\label{sec:strong_defences}
%

Given vectors $Y$ created from corrupting some $\epsilon \cdot n$ in $X$, the aggregator does not know which elements in $Y$ are corrupted. One principled way is to solve the problem is to distinguish {\em inliers} from {\em outliers} in $Y$, and lower the contribution of the latter when computing a measure of  centrality.
Nearly all robust aggregators work with this principle, either explicitly or implicitly, but differ
in how they do so. 

\begin{algorithm}[t]
\caption{Meta-algorithm for strong \defences}
\label{alg:existing_algos}
\begin{algorithmic}[1]
\REQUIRE $\epsilon$-corrupted set $Y = \{y_1, ..., y_n\} \subseteq \mathbb{R}^d$, $n$, $\epsilon$, and $||\Sigma||_2$\\
\ENSURE $\tilde{\mu}$ robust mean 
\STATE $\xi := k \cdot ||\Sigma||_2$ \hspace{0.5cm} $\triangleright$ Choose $\sqrt{20} < k \leq 9$ ~\cite{zhu2023byzantine, diakonikolas2019robust}
\STATE $Y' = Y$
\FOR{$j=0$ {\bfseries to} $j= 2 \cdot n \cdot \epsilon - 1$} 
    \IF{$||\text{Cov}(Y')||_2 \leq \xi$}
        \RETURN $\tilde{\mu} = \frac{1}{n} \sum_{i=1}^ny'_i$
    \ELSE
        \STATE $Y'  \leftarrow $ OutlierRemovalSubroutine($Y'$, $\epsilon$, $||\Sigma||_2$) 
    \ENDIF
\ENDFOR
\RETURN $\tilde{\mu} = \frac{1}{n} \sum_{i=1}^ny'_i$
\end{algorithmic}
\end{algorithm}

\begin{figure*}[t]
\label{fig:3algos}
  \begin{minipage}[t]{.33\textwidth}
\begin{algorithm}[H]
\caption{FILTERING ~\cite{diakonikolas2017being}}
\label{alg:filtering}
\scriptsize
\begin{algorithmic}[1]
\REQUIRE $\epsilon$-corrupted set $Y = \{y_1, ..., y_n\} \subseteq \mathbb{R}^d$, $||\Sigma||_2$  \\
\ENSURE $Y' = \{y'_1, ..., y'_n\}$ updated set $Y$ removing or undermining outliers \\
\STATE $c_i := 1/n, \hspace*{0.25cm} i \in [n]$ \hspace{0.5cm} $\triangleright$ Initialize equal weights
\STATE $\mu_c := \sum_{i = 1}^n c_iy_i$  \hspace{0.95cm} $\triangleright$ Compute mean of $Y$
\STATE $v  \leftarrow$ largest eigenvector of $\text{Cov}(Y)$
\STATE $\tau_{i} := \langle y_i - \mu_{c}, v\rangle ^2, \hspace*{0.25cm} i \in [|Y|]$
\STATE $c_i := c_i(1 - \frac{\tau_{i}}{\tau_{max}}), \hspace*{0.25cm} i \in [|Y|]$ 
\STATE $cnt$ := count vectors in $Y$ with $c_i \neq 0$
\STATE $c_i := c_i / \sum_{i=1}^{cnt} c_i$ \hspace{0.1cm} $\triangleright$ Normalize the new weights  
\RETURN $\{c_1\cdot y_1, \dots, c_n\cdot y_n\}$
\vspace{1.75cm}
\end{algorithmic}
\end{algorithm}
  \end{minipage}%
  \hfill
  \begin{minipage}[t]{.33\textwidth}
\begin{algorithm}[H]
\caption{NO-REGRET ~\cite{hopkins2020robust}}
\label{alg:noregret}
\scriptsize
\begin{algorithmic}[1]
\REQUIRE $\epsilon$-corrupted set $Y = \{y_1, ..., y_n\} \subseteq \mathbb{R}^d$, $||\Sigma||_2$ \\
\ENSURE $Y' = \{y'_1, ..., y'_n\}$ updated set $Y$ removing or undermining outliers \\
\STATE $l_{i,j} = ||y_i - y_j||_2$ for all $i,j \in [n]$
\STATE $\eta := 0.5 / (\text{max}(l_{i,j}))^2$  $\triangleright$ Step size for re-weighting  
\STATE $c_i := 1/n, \hspace*{0.25cm} i \in [n]$ \hspace{0.5cm} $\triangleright$ Initialize equal weights
\STATE $\mu_c := \sum_{i = 1}^n c_iy_i$  \hspace{0.95cm} $\triangleright$ Compute mean of $Y$
\STATE $v  \leftarrow$ largest eigenvector of $\text{Cov}(Y)$
\STATE $\tau_{i} := \langle y_i - \mu_{c}, v\rangle ^2, \hspace*{0.25cm} i \in [|Y|]$
\STATE $c_i := c_i(1 - \tau_{i} \cdot \frac{\epsilon \eta}{2 ||\Sigma||_2 d}), \hspace*{0.25cm} i \in [|Y|]$ 
\STATE $\Delta_{n, \epsilon} := \{\Tilde{C} | \sum \Tilde{c}_i = 1, \Tilde{c}_i \leq \frac{1}{(1 - \epsilon)n}\}$
\STATE $C := \text{arg min}_{\Tilde{C} \in \Delta_{n, \epsilon}} D_{KL}(\Tilde{C}||C)$   
\RETURN $\{c_1 \cdot y_1, \dots, c_n \cdot y_n\}$
\vspace{0.7cm}
\end{algorithmic}
\end{algorithm}
  \end{minipage}%
  \hfill
  \begin{minipage}[t]{.33\textwidth}
\begin{algorithm}[H]
\caption[Caption for LOF]{Sum of Squares (cf.~\cite{kothari2017outlierrobust})}
\scriptsize
\begin{algorithmic}[1]
\REQUIRE $\epsilon$-corrupted set $Y = \{y_1, ..., y_n\} \subseteq \mathbb{R}^d$, $||\Sigma||_2$\\
\ENSURE $Y' = \{y'_1, ..., y'_n\}$ updated set $Y$ removing or undermining outliers \\
\STATE $W := \{w_1, \dots, w_n\}$ $\hspace*{0.05cm} \triangleright$ create variables for weights
\STATE $Y':= \{y'_1, \dots, y'_n\}$ $\hspace*{0.5cm} \triangleright$ create variables for $Y'$
\STATE $\mu' := \frac{1}{n}\sum_{i=1}^{n} y'_i$ $\hspace*{0.6cm} \triangleright$ create a variable for mean
\STATE $\Sigma' := \frac{1}{n}\sum_{i=1}^{n}(y'_i - \mu')^{T} \cdot (y'_i - \mu)$ 
\STATE $\mathbb{\tilde{E}} \leftarrow $ initialize a SoS program with $W$ and $Y'$
\STATE $\mathbb{\tilde{E}}$.add($w_i^2 = w_i$ for every $i \in [n]$) $\triangleright$ add constraints 
\STATE $\mathbb{\tilde{E}}$.add($\sum_{i=1}^{n} w_i = (1 - \epsilon)n$)
\STATE $\mathbb{\tilde{E}}$.add($w_iy'_i = w_iy_i$ for every $i \in [n]$)
\STATE $\mathbb{\tilde{E}}$.add($\frac{1}{n} \sum\langle y'_i - \mu', v\rangle^2 \leq 9 ||\Sigma||_2$) for all $v \in \mathbb{R}^d$
\STATE $\mathbb{\tilde{E}}$.solve() $ \hspace*{0.5cm} \triangleright$ solve the program to get operator $\mathbb{\tilde{E}}$
\RETURN $\{\mathbb{\tilde{E}}(y'_1), \dots, \mathbb{\tilde{E}}(y'_n)\}$
\label{alg:sos}
\end{algorithmic}
\end{algorithm}
  \end{minipage}
  \caption{Pseudocode for OutlierRemovalSubroutine in Alg.~\ref{alg:existing_algos} as instatiated by $3$ different strong \defences.}
  \label{fig:3algos}
\end{figure*}

\paragraph{\Defences~with weak bounds on bias.}
Weak robust aggregators, such as Trimmed Mean and Median, compute measures of centrality {\em per dimension}.
When we compute central tendencies for each dimension separately, an adversary can always create a bias of $\Theta(\sqrt{\epsilon})$ in a dimension, proportional to the benign variance in that dimension, by corrupting $\epsilon$ fraction of $X$. This is true for $X$ sampled from any distribution with bounded variance (see Sec. 3 in \cite{diakonikolas2019robust} and Fact 1.2 in~\cite{diakonikolas2019recent}). 
So, it is difficult to bound the total bias better than $O(\sqrt{\epsilon d})$ by analyzing individual dimensions~\cite{diakonikolas2019robust, lugosi2021robust}. 

Let us take the example of Trimmed mean for illustration. Fig.~\ref{fig:weak-strong-defenses} (middle) shows a corrupted $2$-d sample that induces high bias against it. The attacker can create $\epsilon \cdot n$ outliers positioned at the extreme of benign samples along each dimension, displacing the trimmed mean in each dimension by $\sqrt{\epsilon}$, giving total bias of $O(\sqrt{\epsilon d})$. Recent untargeted poisoning attacks have also experimentally shown that such high biases are practical on weak aggregators~\cite{fang2020local}.

\paragraph{\Defences~with strong bounds on bias.} Strong robust aggregators differ from weak ones in that they look at the magnitude of the vectors in $Y$ along {\em all possible vector directions}, not just their individual components per dimension.
They are able to provide substantially better upper bounds of $\tilde{O}(\sqrt{\epsilon})$ on bias, independent of $d$, because of a key fact: If two sets share at least $(1 - \epsilon)$ fraction of elements and exhibit equal maximum variance, then their respective means cannot deviate by more than $O(\sqrt{\epsilon})$ from each other (see the explanation below Fact $1.2$ in~\cite{diakonikolas2019recent}). 
The spectral norm $||\Sigma||_2$ of the covariance matrix captures the maximum variance across all vector directions (it is the largest eigenvalue). Using this principle, strong \defences~check if the variance of the corrupted inputs $Y$ exceeds the expected $||\Sigma||_2$ known from clean samples. 

The first polynomial-time strong aggregators have been proposed only in the last decade~\cite{diakonikolas2017being,diakonikolas2019recent}. Existing poisoning attacks fail to create much bias against them~\cite{zhu2023byzantine}.

Alg.~\ref{alg:existing_algos} is a common backbone for all polynomial-time strong aggregators. The basic idea is to compute a weighted average of vectors in $Y$. Lower weights are given to outliers. The way the weights are computed varies across different strong aggregators and is abstracted as the OutlierRemovalSubroutine (Line 7). The entire procedure is iterative. In each iteration,
the vectors in $Y$ are re-weighted until the maximum variance of the newly weighted samples falls below 
a threshold $\xi$ (Line 4). The threshold depends on $||\Sigma||_2$ (Line 1). This process ensures that all outliers in $Y$ are assigned minimal weight, rendering the mean of the re-weighted points a reliable estimate of the true mean.

The value $\xi$ depends on the natural variance ($||\Sigma||_2$) of the benign samples, trading off how many inliers vs. outliers the procedure is willing to exclude. If the defense chooses $\xi$ much smaller than the natural variance of benign samples, it will filter out benign samples. For example, a constant fraction (about $0.27$) of benign vectors are statistically expected to be between $||\Sigma||_2^{\frac{1}{2}}$ and $2\cdot ||\Sigma||_2^{\frac{1}{2}}$ for normal (Gaussian) distributions---therefore, defenses which filter out points in this range are likely to have bias at least proportional to $\sqrt{\epsilon}$, for $\epsilon=0.27$, even without any corruption. Therefore, practical defenses would use $\threshold > k \cdot ||\Sigma||_2$, for some constant $k$. Prior works suggest using $k=\sqrt{20}$ ~\cite{zhu2023byzantine} and $k=9$~\cite{diakonikolas2017being}. 
The provable bound on bias then is $\tilde{O}(\sqrt{\epsilon}) \cdot ||\Sigma||^{\frac{1}{2}}_2$~\cite{diakonikolas2017being}.

Fig.~\ref{fig:weak-strong-defenses} (right) plots
the behavior of Alg.~\ref{alg:existing_algos} on the same example that creates a large bias against the weak trimmed mean. The samples outside the radius (dark solid circle) defined by $\xi$ have near zero weight. The resulting estimate is much closer to the true mean than the trimmed mean.

The weight assignment procedure varies across all the different strong aggregators proposed in prior work: Filtering~\cite{diakonikolas2017being}, No-Regret~\cite{hopkins2020robust}, and SoS~\cite{kothari2017outlierrobust}.
Fig.~\ref{fig:3algos} shows the different strategies to implement the OutlierRemovalSubroutine the $3$ aggregators use---all follow the same meta-procedure of Alg.~\ref{alg:existing_algos}. Readers need not understand their details, they are included to make the paper self-contained. We experimentally evaluate them in Section~\ref{sec:eval}, when possible.

\paragraph{Computational bottleneck.} The maximum variance direction of $Y$ (Line 3 in Alg.~\ref{alg:filtering}) is computationally expensive to compute. We will later show in Section~\ref{sec:tradeoff} that it is fundamental to the problem of robust aggregation, not just to the above solutions. 
This computational bottleneck forces practical realizations of Alg.~\ref{alg:existing_algos} to operate on smaller subsets of dimensions at a time, as explained in Section~\ref{sec:our_attack}.


\section{The \attack~Attack}
\label{sec:our_attack}

We propose a principled untargeted poisoning attack to defeat existing polynomial-time strong aggregators.

\subsection{Warm up: Attack in Low Dimensions} 
 
The invariant in strong \defences~(see Alg.~\ref{alg:existing_algos}) is that $Y'$, a re-weighted (scaled) version of $Y$, has maximum variance below a threshold $\threshold$. The attack preserves this invariant: Given uncorrupted vectors $X$, it replaces $n\cdot \epsilon$ of them to construct $Y'$ such that  $\text{Cov}(Y')$ is below the threshold. Thus, all the corrupted vectors will be used in the final weighted mean statistic computed by Alg.~\ref{alg:existing_algos}. The attack procedure is given in Alg.~\ref{alg:our_attack_algo}. It computes the mean vector $\hat{\mu}$ of $X$ (line 1) and creates corrupted vectors that hit the maximum allowable deviation from the mean towards zero, without getting filtered out. 
Specifically, line $5$ computes the exact magnitude of the corruptions to create along the direction of $-\hat{\mu}$ such that after corrupting $n \cdot \epsilon$ vectors (as in line $6$), the resulting $Y$ has maximum variance below $\threshold$.

To achieve that, recall that the threshold $\threshold$  for which the guarantees of the robust aggregator hold, is always higher than the benign sample variance along any direction. Ideally, the threshold should be $\geq 9 \cdot \sigma_{max}^2$ as given by theoretical analyses in prior works~\cite{diakonikolas2017being}. Prior implementations for FILTERING and NO-REGRET~\cite{zhu2023byzantine} use $\sqrt{20}$ times a constant estimate of $\sigma_{max}^2$ as the threshold. We make a key observation that there exists a  direction $\hat{s}$, along which corrupted vectors can shift the aggregate far from the original mean to the maximum value, while ensuring that the variance in all other directions remains below the threshold.
Specifically, using the magnitude $z$ given in Eqn.~(\ref{eq:z}) below increases the bias as much as possible without exceeding the threshold.  
We use the offset $-z\hat{s}$ for all corrupted vectors.

\begin{gather}
    z = \sqrt{\frac{\threshold - \sigma_{max}^2}{\epsilon^2 + \epsilon \cdot (1 - \epsilon)^2}}  - \mu^{\hat{s}}\label{eq:z}\\
    \mu^{\hat{s}} = \frac{1}{n}\sum_{i=0}^{n}\langle x_i, \hat{s} \rangle
\end{gather}

Our analysis shows that the above strategy will not increase variance in any other direction beyond the threshold (see Lemma~\ref{lem:l2} in Section~\ref{sec:analysis}). This is necessary to not trigger the outlier removal procedure.  The magnitude of corruption is thus maximized, subject to the threshold invariant remaining valid. Our analysis presented in Section~\ref{sec:analysis} proves that using the value of $z$ given in Eq~\ref{eq:z}, \attack~hits the theoretical upper bound on the bias. We confirm it experimentally for practical setups as well in Section~\ref{sec:eval}. Fig. \ref{fig:single_chunk_brahm} illustrates the corruption strategy described above, depicting the placement of all corrupted vectors just inside the variance threshold boundary.

\begin{algorithm}[t]
\caption{Pseudo-code describing \attack}
\label{alg:our_attack_algo}
\begin{algorithmic}[1] 
\REQUIRE benign set $X = \{x_1, ..., x_n\}$, fraction of corruption $\epsilon$, variance threshold $\threshold$
\ENSURE $\epsilon$-corrupted set $Y = \{y_1, ..., y_n\}$ 
\STATE $\hat{\mu} := \frac{1}{n} \sum_{i=1}^n x_i$  \hspace{0.3cm} $\triangleright$ compute the mean of $X$
\STATE $\hat{\Sigma} := \frac{1}{n} \sum_{i=1}^n (x_i - \hat{\mu})^{T} \cdot (x_i - \hat{\mu})$
\STATE $\hat{s} := \frac{\hat{\mu}}{||\hat{\mu}||_2}$
\STATE $\sigma_{max}^2 := \frac{\threshold}{\sqrt{20}}$ \hspace{0.65cm} $\triangleright$ estimate $\sigma_{max}^2$ using $\threshold$ 
\STATE $z = \sqrt{\frac{\threshold - \sigma_{max}^2}{\epsilon^2 + \epsilon(1-\epsilon)^2}}$ \hspace{0.1cm}  $\triangleright$ magnitude of corruption along $\hat{s}$ 
\STATE $y_i := \hat{\mu} - \hat{s} \cdot z$ for all $i \in [1, n \cdot \epsilon]$
\STATE $y_i := x_i$  for all $i \in [n \cdot \epsilon + 1, n]$
\RETURN $Y$
\end{algorithmic}
\end{algorithm}

The direction of the perturbation, i.e., along $-\hat{\mu}$, is chosen to reduce the final $\hat{\mu}$ closer to zero, motivated specifically in untargeted poisoning attacks. If the updates in SGD are close to zero, they carry lesser signals about the training samples, effectively disallowing the model from learning.
Other forms of poisoning can use other directions based on their goals, but we focus only on untargeted poisoning.

Note that our attack does {\em not} have bias dependent on $d$ in low dimensions, as expected from prior theoretical analysis of strong \defences. Our attack does close the gap between the best-known attacks and the upper bound, showing the tightness of prior theoretical analyses.

\begin{figure}[t]
    \centering
    \includegraphics[scale=0.75]{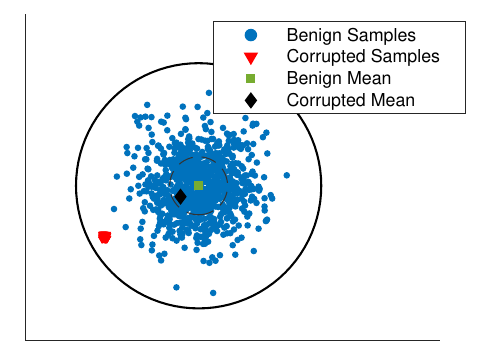}
    \caption{\attack~: Corruptions crafted within variance threshold, yet biasing the mean to an order of $\Omega(\sqrt\epsilon)$}
\label{fig:single_chunk_brahm}
\end{figure}

\subsection{Vulnerability in High Dimensions}

The analysis of strong \defences~makes idealized computational assumptions. The expensive steps in these defenses are computing the maximum variance and a corresponding maximum variance direction of inputs $Y$. Recall that the maximum variance corresponds to the spectral norm of the sample covariance matrix $||\text{Cov}(Y)||_2$, i.e., the largest eigenvalue. The maximum variance direction is the direction of one of the eigenvectors with largest eigenvalue.
The exact computation of $||\text{Cov}(Y)||_2$ is $O(d^3)$ \cite{cuppen1980divide} due to iterative matrix multiplication steps. 
As a rough estimate, it takes about $150$ seconds on a single CPU core of a modern desktop to compute the largest eigenvector sequentially and its eigenvalue for $d$$=$$10^4$. Averaging in SGD is over gradient vectors which can be much larger, for instance, small CNNs need $d=3$ $\times 10^{5}$ and modern large language models can have $d=10^{10}$ to $10^{12}$. Therefore, given finite memory per computational device, the computation over $d \times d$ matrices has to be split into smaller matrices when $d$ is large.



\paragraph{Practical realization of strong \defences.}
To overcome the computational bottleneck, recent works split the dimensions into disjoint chunks that are small enough to be computed on individually~\cite{zhu2023byzantine,shejwalkar2021manipulating}.
We illustrate how practical realization of strong \defences~do this. Let the original dimension of samples be $d$ and FILTERING be the robust aggregator chosen. Then, we take 
the first $m$ dimensions of all vectors as the first chunk,
the next $m$ dimensions of all vectors as the second chunk, and so on. Each vector is partitioned into chunks of size say $m=1000$. This will result in each vector having $c = \lfloor \frac{d}{m} \rfloor$ chunks. The FILTERING aggregator will be run on the first chunk, i.e., first $m$ dimensions of all samples to find the robust aggregate, and similarly on each subsequent chunk. We get $c$ robust aggregates, one for each chunk, which are then concatenated to output the final output vector of dimension $d$. 
This generic chunking strategy is detailed in Algorithm~\ref{alg:brahm}.

\begin{algorithm}[t]
\caption{Practical realizations of strong \defences~in high dimensions} 
\label{alg:brahm}
\begin{algorithmic}[1]
\REQUIRE $\epsilon$-corrupted set $Y = \{y_1, \dots, y_n\} \subseteq \mathbb{R}^d $, partition size $m$, $||\Sigma||_2$, a \text{strong aggregator} $\mathcal{A}$ following Alg.~\ref{alg:existing_algos}
\ENSURE Robust mean of $Y$, $\mu \in \mathbb{R}^d$ 
\STATE $\mu = [0,\ldots,0]$ \hspace{0.5cm} $\triangleright$ initialize with zero vector in $\mathbb{R}^d$
\STATE $i = 0$ 
\WHILE{$i < d$}
    \STATE $\Tilde{Y} \leftarrow$ set of sampled $Y$ using indices in $[i, i + m-1]$
    \STATE $\mu[i, i + m-1] := \mathcal{A}(\tilde{Y}, ||\Sigma||_2)$   $\triangleright$ get chunk aggregate 
    \STATE $i := i + m$
\ENDWHILE
\RETURN $\mu$
\end{algorithmic}
\end{algorithm}

\paragraph{Vulnerability.} 
Chunking introduces a new source of vulnerability.  Running robust aggregation on different chunks independently creates an opportunity for the adversary to bias the result on each chunk separately. Specifically, the adversary can bias the aggregate of each chunk by $\Omega(\sqrt{\epsilon})$ using Algorithm~\ref{alg:our_attack_algo}. Since the final aggregated vector is the concatenation of all of the aggregates of each chunk, the biases from the aggregates from all chunks add up. Therefore, the bias in the final aggregated model is  $\Omega(\sqrt{\epsilon c})$. 

Since the number of chunks $c$ is proportional to $d$, the total bias introduced by \attack~is $\Omega(\sqrt{\epsilon d})\cdot\sqrt{\threshold}$.
Formally, Theorem~\ref{thm:main_theorem} in Section~\ref{sec:analysis} proves our claim. 


\paragraph{Multi-chunk \attack.} Multi-chunk strong \defences~have to decide the threshold to use per chunk. Prior work has used a constant function to bound the variance of each chunk by a fixed constant $\sigma_{max}^2$, which is the maximum value of the variance across all chunks (line 5)~\cite{zhu2023byzantine}. Our attack therefore uses $\sigma_{max}^2$ as a replacement for $||\Sigma||_2$ in Alg.~\ref{alg:our_attack_algo}.  One can extend the defense to consider an adaptive threshold, a different constant multiple of the variance $\sigma_i^2$ for chunk $i$. The \attack~attack would simply use Alg.~\ref{alg:our_attack_algo} for each chunk separately using the respective threshold. The analysis remains the same because the threshold used in each chunk must be proportional to the natural variance $\sigma_i^2$ in that chunk, in order to avoid filtering out inliers. So, the threshold used must be at least a constant times $\sigma_i^2$.

\paragraph{Full vs. Partial Knowledge Setups.}
Algorithm~\ref{alg:our_attack_algo} extends to the partial-knowledge setting as well. The only change is to estimate the $\hat{\mu}$ (see Line $1$) using a subset of benign vectors that the adversary has access to (e.g. its own). 
We have evaluated both setups in Section~\ref{sec:eval}. 

In summary, we have shown \attack, an attack strategy to create a bias of $\Omega(\sqrt{\epsilon d})$ in high dimensions. This defeats the main advantage offered by  strong \defences~over weak ones. Practical incarnations of strong \defences~that give $\tilde{O}(\sqrt{\epsilon})$ bias, therefore, remain elusive.

\subsection{Theoretical Optimality Analysis}
\label{sec:analysis}







Our Theorem~\ref{thm:main_theorem} asserts that \attack~against practical realizations of strong \defences~will result in a bias that will be at least $\Omega(\sqrt{\epsilon d})$. Prior known bias for such aggregators is $\tilde{O}(\sqrt{\epsilon d})$, so our attack is nearly optimal.
Our analysis presented here is for the full-knowledge setting. It remains the same for the partial-knowledge setting, modulo the estimation error in computing $\hat{\mu}$ (Line $1$ in Alg.~\ref{alg:our_attack_algo}), which varies by datasets.

\begin{theorem} \label{thm:main_theorem}
    \attack~, as outlined in Algorithm~\ref{alg:our_attack_algo}, will result in a bias of  $\Omega(\sqrt{\epsilon d})\cdot ||\Sigma||_2^{\frac{1}{2}}$ against Alg.~\ref{alg:brahm} in the worst case.
\end{theorem}

\begin{proof}
    We prove three Lemmas~\ref{lem:l1}, ~\ref{lem:l2}, and~\ref{lem:l3} that are stated below to prove this theorem.  We provide complete proofs for the Lemmas in Appendix~\ref{sec:full_proofs}.
    
    In the Lemma~\ref{lem:l1}, we first prove that along any direction $\hat{s}$, we can corrupt an $\epsilon$ fraction of samples to increase the variance along that direction to the maximum possible value defined by the threshold $\threshold$. We show it suffices to use a magnitude of perturbation close to that given in Eq.~\ref{eq:z} previously, to create $\Omega(\sqrt{\epsilon})\cdot||\Sigma||_2^{\frac{1}{2}}$ bias with high probability.

    In the second Lemma~\ref{lem:l2}, we claim that if the corrupted vectors have magnitude of $z = \sqrt{\frac{\threshold - \sigma_{max}^2}{\epsilon^2 + \epsilon(1-\epsilon)^2}} - \mu^{\hat{s}}$ the maximum variance will be $\threshold$ and the maximum variance direction will be $\hat{s}$. Thus, none of the points will be filtered out (or reweighted) by strong~\defences~(see Alg.~\ref{alg:existing_algos}). 

    It follows from Lemma $1.1$ and $1.2$ that \attack~designs corruptions that do not get filtered out by strong \defences~and create $\Omega(\sqrt{\epsilon})\cdot||\Sigma||_2^{\frac{1}{2}}$ bias per chunk.  

    Finally, Lemma~\ref{lem:l3} proves that if the \defence~runs  independently over $c$ different chunks of original vectors, as done in Alg.~\ref{alg:brahm}, with threshold $\threshold$, then the bias is $\Omega(\sqrt{\epsilon c})\cdot \sqrt{\threshold}$. When $\sqrt{\threshold}\geq ||\Sigma||_2^{\frac{1}{2}}$ and $c = \lfloor\frac{d}{m}\rfloor$ where $m$ is a constant representing the size of each chunk such that $m \ll d$, then the bias is about $\Omega(\sqrt{\epsilon d})\cdot ||\Sigma||_2^{\frac{1}{2}}$.
\end{proof}

\begin{restatable}[]{lemma}{lemla}
\label{lem:l1}
    Let $\hat{s}$ be any direction, $\mu^{\hat{s}}$ be the component of benign mean along $\hat{s}$, and $\sigma_{\hat{s}}$ be the variance along $\hat{s}$. 
    Then, replacing an $\epsilon\cdot n$ vectors in $X$ with $-z (\hat{s})$, 
    where $z = \sqrt{\frac{\threshold - \sigma_{\hat{s}}^2}{\epsilon^2 + \epsilon \cdot (1 - \epsilon)^2}} - \mu^{\hat{s}}$ $\pm \frac{\delta \cdot \sigma_{\hat{s}}}{ \sqrt{n}}$, results in bias of $\Omega(\sqrt{\epsilon})||\Sigma||^{\frac{1}{2}}$ with probability at least $1 - \exp(-\delta^2)$, for all $ \delta > 1$.
\end{restatable}

\begin{restatable}[]{lemma}{lemlb}
\label{lem:l2}
    Let $\hat{s}$ be in the direction of the benign aggregate mean $\mu = \frac{1}{n} \sum_{i=1}^n x_i$ and $Y$ be $\epsilon$-corrupted set of vectors. Let the corrupted vectors be along $-\hat{s}$ with magnitude $z = \sqrt{\frac{\threshold - \sigma_{max}^2}{\epsilon^2 + \epsilon(1-\epsilon)^2}} - \mu^{\hat{s}}$. Then, $||\Sigma_Y||_2 \leq \threshold$, where $||\Sigma_Y||_2$ is the spectral norm of the covariance matrix of $Y$.
\end{restatable}

\begin{restatable}[]{lemma}{lemlc}
\label{lem:l3}
    If strong \defences~are used to robustly aggregate over $c$ chunks, as in Algorithm~\ref{alg:brahm} with a single threshold $\threshold$, then \attack~performed over each chunk, as in Algorithm~\ref{alg:our_attack_algo}, will result in a bias of $\Omega(\sqrt{\epsilon c})\cdot \sqrt{\threshold}$.
\end{restatable}

\begin{mdframed}
\textbf{Elements of Construction (EC)}
 \begin{enumerate}
  \item $\mathcal{D} \leftarrow $ $d$ dimensional spherical Gaussian distribution $(O, I)$
  \item $Y = \{y_1, \ldots, y_{n - n \epsilon}\} \overset{\text{i.i.d.}}{\sim} \mathcal{D}$:
    \begin{enumerate}
        \item $\frac{1}{n} \sum_{i=1}^{n - n\epsilon} y_i = \hat{\mu}$ 
        \item Cov($Y$) $ = \hat{I}$
        where $\hat{\mu} \approx O$ and $\hat{I} \approx I$
    \end{enumerate}
    
  \item $B = \{b_1, b_2\} \leftarrow$ randomly chosen $2$ orthogonal unit vectors
  \item $L_i = \{\sqrt{d} + l_{i1}, \ldots, \sqrt{d} + l_{i\frac{n \epsilon}{2}}\} \leftarrow$ distance of corrupted vectors from the origin along $i$th vector in $B$ such that :
    \begin{enumerate}
        \item $|l_{ij}| \ll \sqrt{d} \hspace*{0.5cm} \forall j \in [1, \frac{n \epsilon}{2}]$
        
            for instance, $|l_{ij}|\in [0,1]$
        \item $\sum_{j=1}^{\frac{n \epsilon}{2}} l_{ij} = 0$
        \item $\sum_{j=1}^{\frac{n \epsilon}{2}} l_{1j}^2 = \sum_{j=1}^{\frac{n \epsilon}{\kappa}} l_{2j}^2 $ 
    \end{enumerate}
  \item $c_1, c_2 \leftarrow$ set of corrupted vectors along each vector in $B$ such that for each $c_i$: \\
    $c_i = \{(\sqrt{d} + l_{i1})\cdot b_i, \ldots, (\sqrt{d} + l_{i\frac{n \epsilon}{2}})\cdot b_i\}$
  \item $C = \{y_{n - n\epsilon + 1}, \ldots, y_n\} \leftarrow c_1 \cup c_2$  
  \item $Y' = Y \cup C$
\end{enumerate} 
\end{mdframed}

\subsection{Is Computational Bottleneck Fundamental?}
\label{sec:tradeoff}

\attack~works against several strong \defences, all of which face the computational bottleneck of computing the maximum variance direction of given vectors. It corresponds to computing the largest eigenvector, a problem of broad interest in ML that does not scale with dimensions. How fundamental is the connection between designing a strong \defence~and computing the maximum variance direction of the given vectors? Specifically, {\em is it necessary for strong \defences~to be as computationally expensive as finding the maximum variance direction?} 

\paragraph{Bias vs. computational complexity.}

We will construct a set of vectors, call it $Y'$, for which computing its maximum variance direction approximately will have complexity not much worse than that of computing the aggregate with a strong \defence~$f$. Theorem~\ref{thm:connection} given later formalizes the claim. To the best of our knowledge, the worst-case time complexity of computing the maximum variance direction approximately\footnote{We seek approximations where the approximation error decreases linearly with the number of vectors $n$. With regards to such error, approximate methods like power iteration operate with a time complexity of $\tilde{O}(n^2d)$ \cite{lecture_notes_shayan}.} for the set $Y'$ is $min(\tilde{O}(n^2d),O(d^3))$~\cite{lecture_notes_shayan ,cuppen1980divide, arbenz2012lecture} and hence, devising more efficient $f$ would imply faster algorithms for the former on $Y'$. We refer to such sets $Y'$ as a {\em non-trivial instances} for computing the maximum variance direction.  To create $Y'$ we rely on an initial set of vectors sampled from a spherical Gaussian \footnote{These are distributions with equal variance along every direction.} distribution as it suffices to show the existence of such sets for which the reduction is valid.



\paragraph{Constructing $Y'$.} We follow the aforementioned elements of construction to construct the set $Y'$. Consider a $d$ dimensional spherical Gaussian distribution with mean at Origin and $I$ as covariance ($1$). Now, sample a set of vectors $Y$, i.i.d from this distribution as shown in ($2$). We construct the set $Y'$ from $Y$  by adding $n\epsilon$ {\em corrupted vectors} to $Y$ following the steps from ($3$) to ($7$). First, we choose a set $B$, as per ($3$), which represents $2$ mutually perpendicular directions. Now we distribute $n\epsilon$ vectors equally along these $2$ directions. The distances of these vectors from the origin in each of the directions in $B$ are given in ($4$). Informally, all vectors are slightly different from each other but are close to $\sqrt{d}$ distance from the origin. Accordingly, we get $\frac{n\epsilon}{2}$ corrupted vectors in each of these directions as given in ($5$).  Finally, we add these vectors to the set $Y$ to create $Y'$ ($6$, $7$).

Note that while constructing $Y'$ we choose the corrupted vectors to be at distances of about $\sqrt{d}$ to make them indistinguishable from the benign vectors of $Y$, since samples from a $d$ dimensional Gaussian are at a distance of $\sqrt{d}$ with very high probability~\cite{diakonikolas2023algorithmic} (also see Figure $2.1$ in~\cite{diakonikolas2023algorithmic} for illustration). Further, the different added $L_i$s make sure that the corrupted vectors are also different from each other with respect to their magnitude ($L_2$ norm). Hence, the corrupted and benign vector magnitudes follow the same distribution.

We first establish in Lemma~\ref{lem:1.1} that the maximum variance of vectors in $Y'$ lies approximately along the difference of vectors in $B$. We provide the proof in Appendix~\ref{sec:full_proofs}. 


\refstepcounter{theorem} 

\begin{restatable}[]{lemma}{lemfa}\label{lem:1.1}
    Given an $\epsilon$-corrupted set of vectors $Y'$, as constructed using the elements of construction (EC), the direction of maximum variance of $Y'$ lies with a small angle of $cos^{-1}\left(\sqrt{1 - O(\frac{\delta^2}{n}})\right)$ from the difference of vectors in set $B$ with probability at least $1 - 2\exp (\frac{-\delta^2}{2})$, for all $\delta > 2$.
\end{restatable}

We see that the approximation error in the angle drops with $n$. The failure probability drops exponentially in $\delta$, which is a tunable constant in the analysis.


\paragraph{Reduction.} Lemma~\ref{lem:1.1} provides a way to compute the maximum variance direction of $Y'$ if we know the difference of vectors in $B$. So, how does one find this difference given $Y'$? We provide a reduction in Alg.~\ref{alg:reduction_algo} to show that the direction of maximum variance of any set $Y'$, constructed as above, can be computed approximately using the outputs of robust aggregator $f(Y')$ in quasilinear extra time $\tilde{\Theta}(|Y'|)$.
Theorem~\ref{thm:connection} presented later formalizes the precise claim.

Alg.~\ref{alg:reduction_algo} works as follows. It computes two vectors $\mu'$, the average of $Y'$ (line $1$) and $\hat{\mu}$, the robust aggregate $f(Y')$ (line $2$). Then, it iterates over all the vectors in $Y'$ and checks for the projection of $\mu'-\hat{\mu}$ along each of them (lines $4$-$6$). The $n\epsilon$ vectors with the largest projections are identified as the set of corrupted vectors $C'$ (line $8$). Next, it iterates over $C'$ to find a pair of mutually perpendicular vectors. Finally, it takes the unit vectors along each of the mutually perpendicular vectors and returns the difference of these unit vectors as output (lines $9$-$16$). 

\begin{algorithm}[t]
\caption{Reduction Algorithm } 
\label{alg:reduction_algo}
\begin{algorithmic}[1]
\REQUIRE $Y' =\{y_1, \dots, y_n\} \subseteq \mathbb{R}^d$, strongly-bounded \defence~algorithm $f$ 
\ENSURE $v^*$, direction of maximum variance
\STATE $\mu' := \frac{1}{n}\sum_{i=1}^{n} y_i$ \hspace{0.95cm} $\triangleright$ compute average of the set 
\STATE $\tilde{\mu} := f(Y')$ \hspace{1cm} $\triangleright$ compute robust mean of the set 
\STATE $S := \{\}$ $\hspace*{3.5cm} \triangleright$ create an empty set 
\FOR{$k=1$ {\bfseries to} $k=n$} 
    \STATE $s_i := \langle \mu' - \tilde{\mu}, \frac{y_i}{||y_i||_2} \rangle$
    \STATE $S \leftarrow s_i$ $\hspace{4.3cm} \triangleright$ add $s_i$ to $S$
\ENDFOR
\STATE $C' \leftarrow$ set of $n \epsilon$ vectors with highest value in $S$
\STATE $c_1 := C'[1]$  \hspace{2.0cm} $\triangleright$ take the first vector of $C'$ 
\FOR{$k=2$ {\bfseries to} $k=n\epsilon$} 
    \STATE $c_2 := C'[k]$
    \IF{$\langle c_1, c_2\rangle = 0$}
        \RETURN $\frac{c_1}{||c_1||_2} - \frac{c_2}{||c_2||_2}$
    \ENDIF
\ENDFOR
\RETURN $\frac{c_1}{||c_1||_2} - \frac{c_2}{||c_2||_2}$
\end{algorithmic}
\end{algorithm}

Lemma~\ref{lem:1.2} below states that Alg.~\ref{alg:reduction_algo} indeed finds the difference of vectors in set $B$ with very high probability. 

\begin{restatable}[]{lemma}{lemfb}\label{lem:1.2}
    Let $f(Y')=\tilde{\mu}$ and $\mu'=\frac{1}{n}\sum_{i=0}^{n}{Y'[i]}$, then Alg.~\ref{alg:reduction_algo} returns the difference of vectors in set $B$ with probability at least $1 - n\cdot \exp( \frac{-n\delta^2}{2})$, for all $\delta > 1$.
\end{restatable}
    
\begin{proofsketch}
    We defer the full proof to Appendix~\ref{sec:full_proofs} and provide a sketch here. We show that the projection of $\mu'-\tilde{\mu}$ is much larger along the corrupted vectors than along others in set $Y'$ with very high probability. Hence, it is sufficient to iterate over vectors in $Y'$ and get the top $n \cdot \epsilon$ vectors in the order of $\mu' - \tilde{\mu}$'s projections along them to find the corrupted ones. Thus, in Line $8$, $C'$ represents the set of corrupted vectors. Since equal numbers of vectors in $C'$ align along each vector in set $B$, if we choose any vector from $C'$ we can find a perpendicular corrupted to it in a single iteration. The difference of these corrupted vectors normalized by their norm is equal to the difference of vectors in set $B$. It implies Alg.~\ref{alg:reduction_algo} returns the difference of vectors in set $B$.  
    \qed
\end{proofsketch}

\addtocounter{theorem}{-1}

\noindent Using Lemmas~\ref{lem:1.1} and~\ref{lem:1.2}, we state and prove the desired Theorem~\ref{thm:connection} below. 
We use the $\tilde{\Theta}$ version of the asymptotic complexity in the proof, which ignores logarithmic factors.

\begin{theorem}\label{thm:connection}
Let $f$ be a strong robust aggregator with running time $T(f)$ and $Y'$ be a set of vectors constructed as in EC. Then, the direction of maximum variance of $Y'$ can be computed with a small approximation error, in time $T(f) + \tilde{\Theta}(|Y'|)$, with very high probability.
\end{theorem}

\begin{proof}
Given the set $Y'$,  Lemma ~\ref{lem:1.1} asserts that the direction of maximum variance of $Y'$ lies approximately along the difference of vectors in set $B$. The approximation error in the angle is stated in Lemma ~\ref{lem:1.1} and approaches zero with increasing number of samples $n$. Lemma ~\ref{lem:1.2} asserts that the reduction stated in Algorithm~\ref{alg:reduction_algo} finds the difference of vectors in set $B$ using $Y'$. The total failure probability in these two Lemmas is either exponentially small in $n$ or in $\delta$, so taking a union bound, the final result is correct with very high probability. 
The time taken by Algorithm~\ref{alg:reduction_algo}, excluding Line 2, is $\tilde{\Theta}(n\cdot d)$, which can be verified by inspection. Each vector summation and dot product $\langle \cdot \rangle$ is O($d$). The $2$ loops run $O(n)$ times, with each iteration O($d$). Picking the highest $n\epsilon$ values in $S$ (Line 8) takes $\tilde{\Theta}(n\cdot d)$, which is the same as $\tilde{\Theta}(|Y'|)$. The Line 2 takes time $T(f)$. So, the total running time of Alg.~\ref{alg:reduction_algo} is thus $T(f) + \tilde{\Theta}(|Y'|)$.
\end{proof}

We have shown that the existence of efficient strong aggregators $f$ implies being able to compute the approximate direction of maximum variance, for non-trivial inputs as in EC.
The latter problem, to the best of our knowledge however, has $O(d^3)$ algorithms for general inputs. We are not aware of any faster solutions for the input class EC.

\section{Evaluation}
\label{sec:eval}
\paragraph{Goals.} Our experimental evaluation of machine learning tasks aims to answer two primary questions:
\begin{enumerate}
    \item Does the bias introduced by \attack~during training match the theoretical bias our analysis expects?

    \item How effective is \attack~as an untargeted poisoning attack, i.e., what drop in model accuracy does it induce when used at each training step?
\end{enumerate}

We evaluate on standard image classification datasets, training with SGD using state-of-the-art strong \defences~with strong bias bounds. We evaluate both the partial and full knowledge settings. Furthermore, we aim to compare the impact of our attack on training accuracy with that of existing attacks in these scenarios.

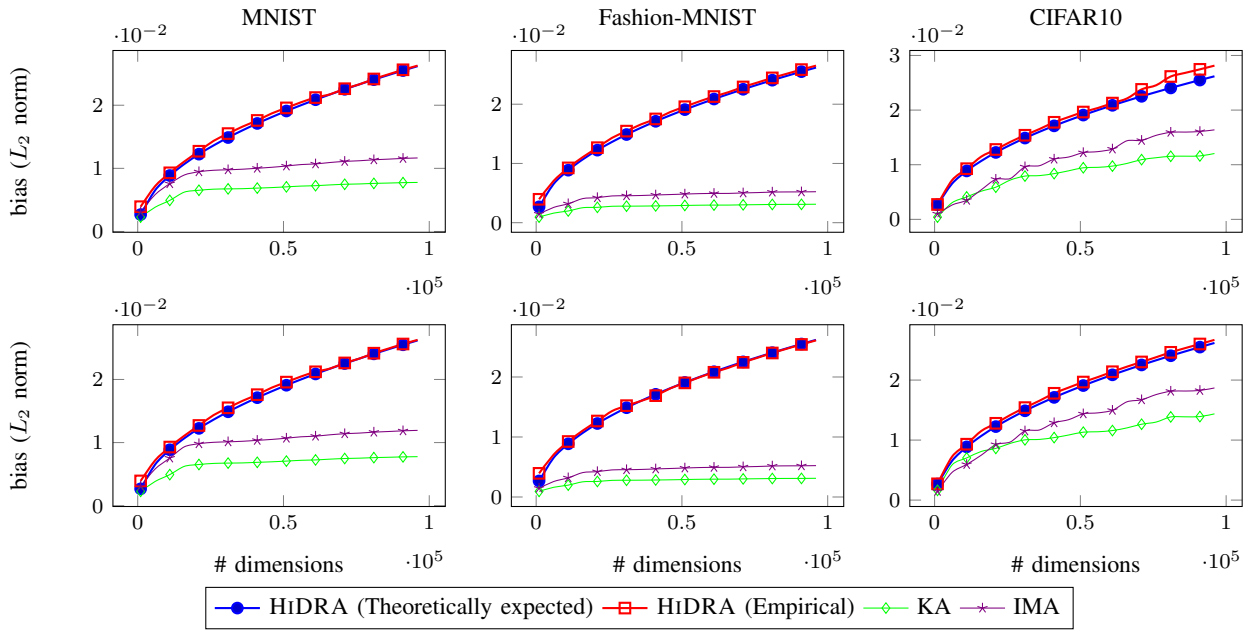
\begin{figure*}[h!]
    \centering
  \begin{tikzpicture}
    \begin{groupplot}[
group style = {group size = 3 by 2, horizontal sep = 25pt, vertical sep = 35pt, xlabels at=edge bottom, ylabels at=edge bottom},
      width=6.0cm,
      height=4.0cm,
    ]

    \nextgroupplot[
      title={MNIST},
      ylabel={bias ($L_2$ norm)},
      grid style=dashed,
      legend style={column sep = 1pt, legend columns = 5, legend to name = grouplegend1, font=\small},
      tick label style={font=\footnotesize}, 
      label style={font=\small},  
      title style={font=\small},
    ]
    
    \addplot+[smooth, mark=*, mark repeat=2, thick] table [y=y1, x=x, col sep=comma] {figs/empirical_bias_filtering_mnist.csv};

    \addplot+[smooth, mark=square, mark repeat=2, thick] table [y=y2, x=x, col sep=comma] {figs/empirical_bias_filtering_mnist.csv}; 

    \addplot+[smooth , mark=diamond, mark repeat=2, color=green] table [y=y3, x=x, col sep=comma]
    {figs/empirical_bias_filtering_mnist.csv};

    \addplot+[smooth , mark=star,  mark repeat=2, color=violet] table [y=y4, x=x, col sep=comma]
    {figs/empirical_bias_filtering_mnist.csv};

    \legend{\attack~(Theoretically expected) , \attack~(Empirical), KA, IMA}
    \nextgroupplot[
      title={Fashion-MNIST},
      grid style=dashed,
      tick label style={font=\footnotesize}, 
      label style={font=\small},  
      title style={font=\small},
    ]

    \addplot+[smooth, mark=*, mark repeat=2, thick] table [y=y1, x=x, col sep=comma] {figs/empirical_bias_filtering_fmnist.csv};

    \addplot+[smooth, mark=square, mark repeat=2, thick] table [y=y2, x=x, col sep=comma] {figs/empirical_bias_filtering_fmnist.csv}; 

    \addplot+[smooth , mark=diamond, mark repeat=2, color=green] table [y=y3, x=x, col sep=comma]
    {figs/empirical_bias_filtering_fmnist.csv};

    \addplot+[smooth , mark=star,  mark repeat=2, color=violet] table [y=y4, x=x, col sep=comma]
    {figs/empirical_bias_filtering_fmnist.csv};

    \nextgroupplot[
      title={CIFAR10},
      grid style=dashed,
      tick label style={font=\footnotesize}, 
      label style={font=\small},  
      title style={font=\small},
    ]

    \addplot+[smooth, mark=*, mark repeat=2, thick] table [y=y1, x=x, col sep=comma] {figs/empirical_bias_filtering_cifar.csv};

    \addplot+[smooth, mark=square, mark repeat=2, thick] table [y=y2, x=x, col sep=comma] {figs/empirical_bias_filtering_cifar.csv}; 

    \addplot+[smooth , mark=diamond, mark repeat=2, color=green] table [y=y3, x=x, col sep=comma]
    {figs/empirical_bias_filtering_cifar.csv};

    \addplot+[smooth , mark=star,  mark repeat=2, color=violet] table [y=y4, x=x, col sep=comma]
    {figs/empirical_bias_filtering_cifar.csv};

    \nextgroupplot[
      grid style=dashed,
      xlabel={\# dimensions},
      ylabel={bias ($L_2$ norm)},
      legend style={column sep = 1pt, legend columns = 5, legend to name = grouplegend1, font=\small},
      tick label style={font=\footnotesize}, 
      label style={font=\small},  
      title style={font=\small},
    ]
    
    \addplot+[smooth, mark=*, mark repeat=2, thick] table [y=y1, x=x, col sep=comma] {figs/empirical_bias_noregret_mnist.csv};

    \addplot+[smooth, mark=square, mark repeat=2, thick] table [y=y2, x=x, col sep=comma] {figs/empirical_bias_noregret_mnist.csv}; 

    \addplot+[smooth , mark=diamond, mark repeat=2, color=green] table [y=y3, x=x, col sep=comma]
    {figs/empirical_bias_noregret_mnist.csv};

    \addplot+[smooth , mark=star,  mark repeat=2, color=violet] table [y=y4, x=x, col sep=comma]
    {figs/empirical_bias_noregret_mnist.csv};

    \legend{\attack~(Theoretically expected) , \attack~(Empirical), KA, IMA}
    \nextgroupplot[
      xlabel={\# dimensions},
      grid style=dashed,
      tick label style={font=\footnotesize}, 
      label style={font=\small},  
      title style={font=\small},
    ]

    \addplot+[smooth, mark=*, mark repeat=2, thick] table [y=y1, x=x, col sep=comma] {figs/empirical_bias_noregret_fmnist.csv};

    \addplot+[smooth, mark=square, mark repeat=2, thick] table [y=y2, x=x, col sep=comma] {figs/empirical_bias_noregret_fmnist.csv}; 

    \addplot+[smooth , mark=diamond, mark repeat=2, color=green] table [y=y3, x=x, col sep=comma]
    {figs/empirical_bias_noregret_fmnist.csv};

    \addplot+[smooth , mark=star,  mark repeat=2, color=violet] table [y=y4, x=x, col sep=comma]
    {figs/empirical_bias_noregret_fmnist.csv};

    \nextgroupplot[
      xlabel={\# dimensions},
      grid style=dashed,
      tick label style={font=\footnotesize}, 
      label style={font=\small},  
      title style={font=\small},
    ]

    \addplot+[smooth, mark=*, mark repeat=2, thick] table [y=y1, x=x, col sep=comma] {figs/empirical_bias_noregret_cifar.csv};

    \addplot+[smooth, mark=square, mark repeat=2, thick] table [y=y2, x=x, col sep=comma] {figs/empirical_bias_noregret_cifar.csv}; 

    \addplot+[smooth , mark=diamond, mark repeat=2, color=green] table [y=y3, x=x, col sep=comma]
    {figs/empirical_bias_noregret_cifar.csv};

    \addplot+[smooth , mark=star,  mark repeat=2, color=violet] table [y=y4, x=x, col sep=comma]
    {figs/empirical_bias_noregret_cifar.csv};
    \end{groupplot}
    \node at (7, -5.0) {\ref{grouplegend1}};
  \end{tikzpicture}
  \caption{Bias vs. \# of Dimensions against FILTERING (top) and NO-REGRET (bottom) strong aggregators.}
  \label{fig:emp_bias}
\end{figure*}

\subsection{Experimental Setup}
\label{sec:eval_setup}

Untargeted poisoning attacks \cite{fang2020local, xie2020fall} are most often evaluated in federated learning setups, which consist of several untrusted clients and a trusted server. Clients train locally and send local model updates (vector) which are aggregated by the server as discussed in Section~\ref{sec:backandprob}. We evaluate in this setup assuming an $\epsilon$ fraction of the clients are malicious and send update vectors created using \attack.

\paragraph{Datasets.} We use $3$ image classification datasets: MNIST ~\cite{lecun1998gradient}, Fashion-MNIST ~\cite{xiao2017fashion}, and CIFAR10 ~\cite{krizhevsky2009learning}. Each of these datasets comprises $60,000$ training and $10,000$ test examples split across $10$ classes. In MNIST and Fashion-MNIST, the examples are $28 \times 28$ grayscale images, while in CIFAR10, they are $32 \times 32$ color (RGB) images.

\paragraph{Models and training.} We use convolutional neural networks (CNN) with two convolutional layers followed by two fully connected layers with ReLU activations. The CNNs used for MNIST, Fashion-MNIST, and CIFAR10 have $3.2\times10^5$, $3.2\times10^5$, and $1.7\times10^6$ parameters respectively, which are the dimensions $d$ of the vectors being aggregated. In all experiments, the data is independently and identically distributed across $100$ federated clients. We fix $\epsilon=0.2$ for all of our experiments to match the setup with the prior work that proposes strong \defences~\cite{zhu2023byzantine}. We also report on varying $\epsilon$ at the end of this section. In each round of federated learning, every client runs $5$ local epochs with batch size $10$ before submitting the updated local model to the server. We fix the learning rate at $0.001$ for MNIST and Fashion-MNIST, training the models for $100$ rounds. In the case of CIFAR10, we use a learning rate of $0.01$ and extend the training to $400$ rounds.
We design and implement all our experiments using Python and PyTorch. All the experiments are conducted on Ubuntu 20.04 LTS servers with $64$ AMD Ryzen Threadripper $3970$X CPUs, $96$G RAM, and $2$ NVIDIA GeForce RTX $3090$. The code for evaluation is provided in~\cite{hidracode}.
%

\paragraph{Strong \defences.} We evaluate all strong \defences~that are tractable to run on high-dimensional vectors. These include FILTERING and NO-REGRET. We follow the setup identical to prior work, where the chunk size is $1000$, the threshold $\threshold$ per chunk is $\sqrt{20}\times10^{-5}$, and the estimated upper bound for $\sigma_{max}^2$ is $10^{-5}$ for each chunk~\cite{zhu2023byzantine}. 
We cannot run the third known strong \defence, called SoS, since it requires solving a system of degree-$4$ equations using a semi-definite programming (SDP) solver. Current SDP solvers cannot do so in a reasonable time for $d>10$, so it is intractable for our tasks.

\paragraph{Prior attacks.} There are $3$ untargeted poisoning attacks considered in prior work evaluating strong aggregators~\cite{zhu2023byzantine}: Krum Attack (KA) ~\cite{fang2020local}, Trimmed Mean Attack (TMA) ~\cite{fang2020local} and Inner-Product Manipulation Attack (IMA)~\cite{xie2020fall}. We
report the bias, and drop in model accuracy, achieved by KA and IMA
against all tractable defenses we evaluate. 

We evaluated TMA for $20$ training iterations on all defenses on all datasets, but we do not report on TMA in detail here. This is because all strong aggregators run very effectively against TMA, but take too much time to run in full. Specifically, we verified that the defenses filter out (assign zero weight to) {\em all the poisoned gradients} created by TMA---unlike for other attacks---in all iterations we tested, so the bias induced by TMA is zero. But such filtering by defenses is iterative and removes each of the $n \cdot \epsilon$ poisoned gradients one at a time, taking $O(d^3)$ time for each. This makes the training very slow (and pointless) to run in full.




\subsection{Empirical vs. Theoretical Bias}
\label{sec:eval_empirical_bias}

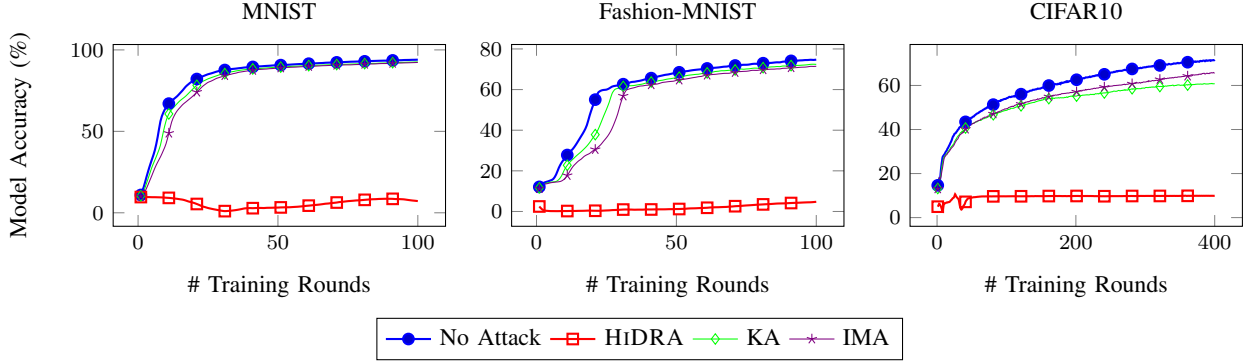
\begin{figure*}[h!]
    \centering
  \begin{tikzpicture}
    \begin{groupplot}[
group style = {group size = 3 by 1, horizontal sep = 25pt, vertical sep = 35pt, xlabels at=edge bottom, ylabels at=edge bottom},
      width=6cm,
      height=4cm,
    ]

    \nextgroupplot[
      title={MNIST},
      xlabel={\# Training Rounds},
      ylabel={Model Accuracy (\%)},
      grid style=dashed,
      legend style={column sep = 1pt, legend columns = 5, legend to name = grouplegend1, font=\small},
      tick label style={font=\footnotesize}, 
      label style={font=\small},  
      title style={font=\small},
    ]
    
    \addplot+[smooth, mark=*, mark repeat=10, thick] table [y=y1, x=x, col sep=comma] {figs/filtering_training_mnist_orig_threshold.csv};

    \addplot+[smooth, mark=square, mark repeat=10, thick] table [y=y2, x=x, col sep=comma] {figs/filtering_training_mnist_orig_threshold.csv}; 


    \addplot+[smooth , mark=diamond,  mark repeat=10, color=green] table [y=y4, x=x, col sep=comma]
    {figs/filtering_training_mnist_orig_threshold.csv};

    \addplot+[smooth , mark=star, mark repeat=10, color=violet] table [y=y5, x=x, col sep=comma]
    {figs/filtering_training_mnist_orig_threshold.csv};

    \legend{No Attack, \attack, KA, IMA}
    \nextgroupplot[
      title={Fashion-MNIST},
      grid style=dashed,
      xlabel={\# Training Rounds},
      tick label style={font=\footnotesize}, 
      label style={font=\small},  
      title style={font=\small},
    ]

    \addplot+[smooth, mark=*, mark repeat=10, thick] table [y=y1, x=x, col sep=comma] {figs/filtering_training_fmnist_orig_threshold.csv};

    \addplot+[smooth, mark=square, mark repeat=10, thick] table [y=y2, x=x, col sep=comma] {figs/filtering_training_fmnist_orig_threshold.csv}; 


    \addplot+[smooth , mark=diamond,  mark repeat=10, color=green] table [y=y4, x=x, col sep=comma]
    {figs/filtering_training_fmnist_orig_threshold.csv};

    \addplot+[smooth , mark=star, mark repeat=10, color=violet] table [y=y5, x=x, col sep=comma]
    {figs/filtering_training_fmnist_orig_threshold.csv};
    \nextgroupplot[
      title={CIFAR10},
      grid style=dashed,
      xlabel={\# Training Rounds},
      tick label style={font=\footnotesize}, 
      label style={font=\small},  
      title style={font=\small},
    ]

    \addplot+[smooth, mark=*, mark repeat=40, thick] table [y=y1, x=x, col sep=comma] {figs/filtering_training_cifar_orig_threshold.csv};

    \addplot+[smooth, mark=square, mark repeat=40, thick] table [y=y2, x=x, col sep=comma] {figs/filtering_training_cifar_orig_threshold.csv}; 


    \addplot+[smooth , mark=diamond,  mark repeat=40, color=green] table [y=y4, x=x, col sep=comma]
    {figs/filtering_training_cifar_orig_threshold.csv};

    \addplot+[smooth , mark=star, mark repeat=40, color=violet] table [y=y5, x=x, col sep=comma]
    {figs/filtering_training_cifar_orig_threshold.csv};
    \end{groupplot}
    \node at (7, -1.5) {\ref{grouplegend1}};
  \end{tikzpicture}
  \caption{Impact of \attack~on accuracy against FILTERING}
  \label{fig:orig_thres_filtering}
\end{figure*}
\begin{figure}[h!]
    \centering
  \begin{tikzpicture}
    \begin{groupplot}[
group style = {group size = 1 by 2, horizontal sep = 25pt, vertical sep = 35pt, xlabels at=edge bottom, ylabels at=edge bottom},
      width=6cm,
      height=4cm,
    ]

    \nextgroupplot[
      title={MNIST},
      ylabel={Model Accuracy (\%)},
      grid style=dashed,
      legend style={column sep = 1pt, legend columns = 5, legend to name = grouplegend1, font=\small},
      tick label style={font=\footnotesize}, 
      label style={font=\small},  
      title style={font=\small},
    ]
    
    \addplot+[smooth, mark=*, mark repeat=10, thick] table [y=y1, x=x, col sep=comma] {figs/noregret_training_mnist_orig_threshold.csv};

    \addplot+[smooth, mark=square, mark repeat=10, thick] table [y=y2, x=x, col sep=comma] {figs/noregret_training_mnist_orig_threshold.csv}; 


    \addplot+[smooth , mark=diamond,  mark repeat=10, color=green] table [y=y4, x=x, col sep=comma]
    {figs/noregret_training_mnist_orig_threshold.csv};

    \addplot+[smooth , mark=star, mark repeat=10, color=violet] table [y=y5, x=x, col sep=comma]
    {figs/noregret_training_mnist_orig_threshold.csv};

    \legend{No Attack, \attack, KA, IMA}
    \nextgroupplot[
      title={Fashion-MNIST},
      xlabel={\# Training Rounds},
      ylabel={Model Accuracy (\%)},
      grid style=dashed,
      tick label style={font=\footnotesize}, 
      label style={font=\small},  
      title style={font=\small},
    ]

    \addplot+[smooth, mark=*, mark repeat=10, thick] table [y=y1, x=x, col sep=comma] {figs/noregret_training_fmnist_orig_threshold.csv};
    
    \addplot+[smooth, mark=square, mark repeat=10, thick] table [y=y2, x=x, col sep=comma] {figs/noregret_training_fmnist_orig_threshold.csv}; 


    \addplot+[smooth , mark=diamond,  mark repeat=10, color=green] table [y=y4, x=x, col sep=comma]
    {figs/noregret_training_fmnist_orig_threshold.csv};

    \addplot+[smooth , mark=star,  mark repeat=10, color=violet] table [y=y5, x=x, col sep=comma]
    {figs/noregret_training_fmnist_orig_threshold.csv};


    



    \end{groupplot}
    \node at (2.25, -5.2) {\ref{grouplegend1}};
  \end{tikzpicture}
  \caption{Impact of \attack~ on accuracy against NO-REGRET}
   \label{fig:orig_thres_noregret}
\end{figure}
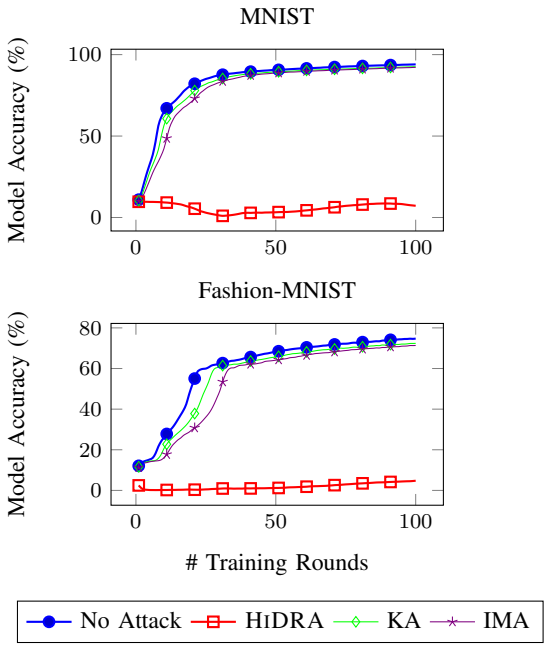

In Section~\ref{sec:analysis}, we prove that \attack~will result in a bias of at least $\sqrt{\epsilon c}\sqrt{\frac{\threshold}{\epsilon + (1-\epsilon)^2}}$. Since the chunk size is $1000$, the theoretically anticipated bias is $\frac{\sqrt{\epsilon d}}{\sqrt{1000}}\cdot\sqrt{\frac{\threshold}{\epsilon + (1-\epsilon)^2}}$. 

We plot the empirical and the theoretical bias using \attack~on all the $3$ datasets against FILTERING and NO-REGRET defenses in Fig.~\ref{fig:emp_bias}. Specifically, for each configuration, consisting of a dataset and a defense, we plot the bias introduced by \attack~for increasing number of dimensions. The empirical bias created by prior attacks KA and IMA attacks is also shown. We choose an arbitrary training round number $10$ to plot the bias, however, the observations are largely the same for all training rounds we sampled. 



The empirical efficacy of \attack~aligns almost exactly with the theoretical bias anticipated. This shows that our analysis is tight. Our result also shows that when \attack~is used against the state-of-the-art strong \defences~the bias increases proportionately to $\sqrt{d}$, confirming our  main claim. Further, observe that the bias introduced by other attacks is much smaller than \attack, highlighting the drastic efficacy gains achieved. The efficacy of prior attacks also increases marginally as $d$, since they introduce some bias in each chunk that adds up across all chunks. But the difference in efficiency between \attack~and prior attacks widens sharply as $d$ increases.

\begin{mdframed}[backgroundcolor=mygray]
    \attack~induces bias proportional to $\sqrt{\epsilon d}$ against strong \defences, showing a contrast to the idealized analysis of these aggregators. Prior attacks have efficacy well below that of our attack. 
\end{mdframed}


\subsection{Impact of \attack~on Accuracy}
\label{sec:eval_accuracy}


The main motivation for untargeted poisoning attacks is to reduce the overall performance of the trained model by introducing bias. Therefore, we also measure how the classification accuracy of the models gets affected after end-to-end training. We find that \attack~has a sharp and negative effect on the model performance, unlike prior attacks which do not affect the model performance much at all when trained with strong \defences.

First, we report the results for the full-knowledge setting and later report the results for the partial-knowledge setting.

\paragraph{\attack~on FILTERING.} Fig.~\ref{fig:orig_thres_filtering} captures the performance of all the considered attacks against FILTERING. \attack~causes the accuracy of CNN on the MNIST dataset to drop by $\mathbf{87}\%$ (from $\mathbf{94}$ to $\mathbf{7}\%$). Similarly, on the Fashion-MNIST dataset and CIFAR10 datasets, the accuracy drops by $\mathbf{70}\%$ and $\mathbf{62}\%$. Existing attacks do not affect the accuracy much. We verify that this is because corrupted vectors created by prior attacks are filtered out by strong \defence~defenses. In contrast, none of the added corrupted vectors added by \attack~are filtered out.
\begin{figure*}
    \centering
  \begin{tikzpicture}
    \begin{groupplot}[
group style = {group size = 3 by 1, horizontal sep = 25pt, vertical sep = 35pt, xlabels at=edge bottom, ylabels at=edge bottom},
      width=6cm,
      height=4cm,
    ]

    \nextgroupplot[
      title={MNIST},
      xlabel={\# Training Rounds},
      ylabel={Model Accuracy (\%)},
      grid style=dashed,
      legend style={column sep = 1pt, legend columns = 2, legend to name = grouplegend1, font=\small},
      tick label style={font=\footnotesize}, 
      label style={font=\small},  
      title style={font=\small},
    ]
    
    \addplot+[smooth, mark=*, mark repeat=10, thick] table [y=y1, x=x, col sep=comma] {figs/filtering_training_mnist_partial_orig_threshold.csv};

    \addplot+[smooth, mark=square, mark repeat=10, thick] table [y=y2, x=x, col sep=comma] {figs/filtering_training_mnist_partial_orig_threshold.csv}; 

    \legend{No Attack, \attack}
    \nextgroupplot[
      title={Fashion-MNIST},
       xlabel={\# Training Rounds},
      grid style=dashed,
      tick label style={font=\footnotesize}, 
      label style={font=\small},  
      title style={font=\small},
    ]

    \addplot+[smooth, mark=*, mark repeat=10, thick] table [y=y1, x=x, col sep=comma] {figs/filtering_training_fmnist_partial_orig_threshold.csv};
    
    \addplot+[smooth, mark=square, mark repeat=10, thick] table [y=y2, x=x, col sep=comma] {figs/filtering_training_fmnist_partial_orig_threshold.csv};

    \nextgroupplot[
      title={CIFAR10},
      grid style=dashed,
       xlabel={\# Training Rounds},
      tick label style={font=\footnotesize}, 
      label style={font=\small},  
      title style={font=\small},
    ]

    \addplot+[smooth, mark=*, mark repeat=40, thick] table [y=y1, x=x, col sep=comma] {figs/filtering_training_cifar_partial_orig_threshold.csv};
    
    \addplot+[smooth, mark=square, mark repeat=40, thick] table [y=y2, x=x, col sep=comma] {figs/filtering_training_cifar_partial_orig_threshold.csv}; 

    \end{groupplot}
    \node at (7, -1.5) {\ref{grouplegend1}};
  \end{tikzpicture}
  \caption{Impact of \attack~(Partial Knowledge) on accuracy against FILTERING.} 
  \label{fig:partial_filtering}
\end{figure*}
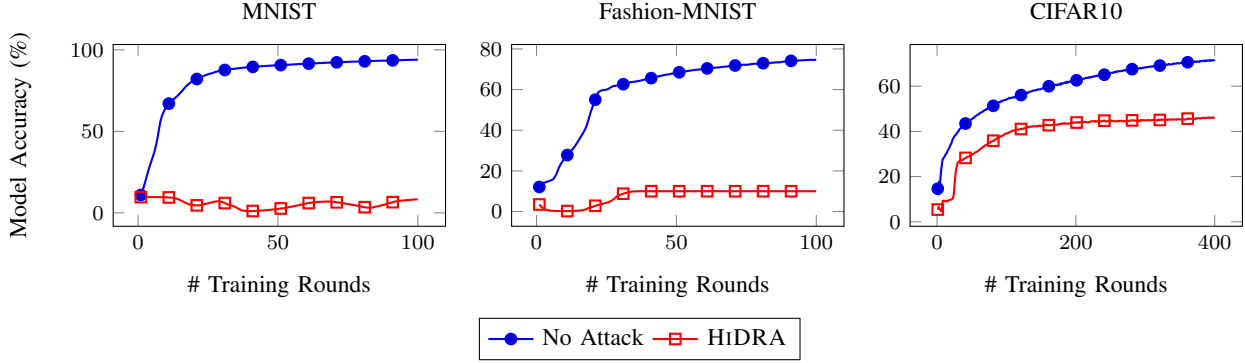
\begin{figure}
    \centering
  \begin{tikzpicture}
    \begin{groupplot}[
group style = {group size = 1 by 2, horizontal sep = 25pt, vertical sep = 35pt, xlabels at=edge bottom, ylabels at=edge bottom},
      width=6cm,
      height=4cm,
    ]

    \nextgroupplot[
      title={MNIST},
      ylabel={Model Accuracy (\%)},
      grid style=dashed,
      legend style={column sep = 1pt, legend columns = 2, legend to name = grouplegend1, font=\small},
      tick label style={font=\footnotesize}, 
      label style={font=\small},  
      title style={font=\small},
    ]
    
    \addplot+[smooth, mark=*, mark repeat=10, thick] table [y=y1, x=x, col sep=comma] {figs/noregret_training_mnist_partial_orig_threshold.csv};

    \addplot+[smooth, mark=square, mark repeat=10, thick] table [y=y2, x=x, col sep=comma] {figs/noregret_training_mnist_partial_orig_threshold.csv}; 

    \legend{No Attack, \attack}
    \nextgroupplot[
      title={Fashion-MNIST},
      xlabel={\# Training Rounds},
      ylabel={Model Accuracy (\%)},
      grid style=dashed,
      tick label style={font=\footnotesize}, 
      label style={font=\small},  
      title style={font=\small},
    ]

    \addplot+[smooth, mark=*, mark repeat=10, thick] table [y=y1, x=x, col sep=comma] {figs/noregret_training_fmnist_partial_orig_threshold.csv};
    
    \addplot+[smooth, mark=square, mark repeat=10, thick] table [y=y2, x=x, col sep=comma] {figs/noregret_training_fmnist_partial_orig_threshold.csv};


    

    \end{groupplot}
    \node at (2.00, -5.2) {\ref{grouplegend1}};
  \end{tikzpicture}
  \caption{Impact of \attack~(Partial-Knowledge) on accuracy against NO-REGRET.}
  \label{fig:partial_noregret}
\end{figure}
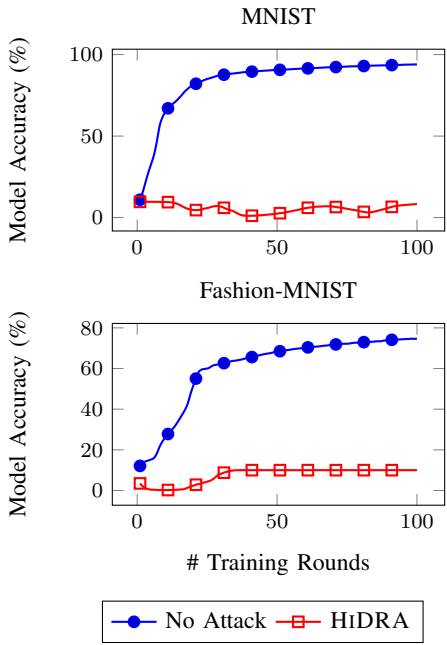
\paragraph{\attack~on NO-REGRET.} 
Fig.~\ref{fig:orig_thres_noregret} illustrates the impact of \attack~on the model performance against NO-REGRET. \attack~ drastically reduces accuracy by $\mathbf{87\%}$ and $\mathbf{70\%}$ for MNIST and Fashion-MNIST respectively.
NO-REGRET is computationally $d$ times more expensive than FILTERING, hence, it is not practical to evaluate it on the CIFAR10 dataset which has at least $5\times$ higher $d$ than the other one. It would take several days to train one model using NO-REGRET on CIFAR10 using our testbed. Therefore, we report results on MNIST and Fashion-MNIST only. 


\begin{mdframed}[backgroundcolor=mygray]
    \attack~completely destroys the ML model performance for all datasets against \defence~defenses we evaluated in the full-knowledge setting.
\end{mdframed}

\begin{figure*}[h]
    \centering
  \begin{tikzpicture}
    \begin{groupplot}[
group style = {group size = 3 by 1, horizontal sep = 25pt, vertical sep = 35pt, xlabels at=edge bottom, ylabels at=edge bottom},
      width=6cm,
      height=4cm,
    ]

    \nextgroupplot[
      title={MNIST},
      xlabel={\# Training Rounds},
      ylabel={Model Accuracy (\%)},
      grid style=dashed,
      legend style={column sep = 1pt, legend columns = 5, legend to name = grouplegend1, font=\small},
      tick label style={font=\footnotesize}, 
      label style={font=\small},  
      title style={font=\small},
    ]
    
    \addplot+[smooth, mark=*, mark repeat=10, thick] table [y=y1, x=x, col sep=comma] {figs/training_mnist_varying_mal_num.csv};

    \addplot+[smooth, mark=star, mark repeat=10, thick, color=violet] table [y=y2, x=x, col sep=comma] {figs/training_mnist_varying_mal_num.csv}; 

    \addplot+[smooth, mark=triangle, mark repeat=10, color=orange, thick] table [y=y3, x=x, col sep=comma]
    {figs/training_mnist_varying_mal_num.csv};

    \addplot+[smooth , mark=diamond, mark repeat=10, color=green, thick] table [y=y4, x=x, col sep=comma]
    {figs/training_mnist_varying_mal_num.csv};

    \addplot+[smooth , mark=square, mark repeat=10, color=red, thick] table [y=y5, x=x, col sep=comma]
    {figs/training_mnist_varying_mal_num.csv};

    \legend{No Attack, 0.01, 0.05, 0.10, 0.20}
    \nextgroupplot[
      title={Fashion-MNIST},
      xlabel={\# Training Rounds},
      grid style=dashed,
      tick label style={font=\footnotesize}, 
      label style={font=\small},  
      title style={font=\small},
    ]

    \addplot+[smooth, mark=*, mark repeat=10, thick] table [y=y1, x=x, col sep=comma] {figs/training_fmnist_varying_mal_num.csv};

    \addplot+[smooth, mark=star, mark repeat=10, thick, color=violet] table [y=y2, x=x, col sep=comma] {figs/training_fmnist_varying_mal_num.csv}; 

    \addplot+[smooth , mark=triangle, mark repeat=10, color=orange, thick] table [y=y3, x=x, col sep=comma]
    {figs/training_fmnist_varying_mal_num.csv};

    \addplot+[smooth, mark=diamond, mark repeat=10, color=green, thick] table [y=y4, x=x, col sep=comma]
    {figs/training_fmnist_varying_mal_num.csv};

    \addplot+[smooth, mark=square, mark repeat=10, thick, color=red] table [y=y5, x=x, col sep=comma]
    {figs/training_fmnist_varying_mal_num.csv};

    \nextgroupplot[
      title={CIFAR10},
      xlabel={\# Training Rounds},
      grid style=dashed,
      tick label style={font=\footnotesize}, 
      label style={font=\small},  
      title style={font=\small},
    ]

    \addplot+[smooth, mark=*, mark repeat=40, thick] table [y=y1, x=x, col sep=comma] {figs/training_cifar_varying_mal_num.csv};

    \addplot+[smooth, mark=star, mark repeat=40, thick, color=violet] table [y=y2, x=x, col sep=comma] {figs/training_cifar_varying_mal_num.csv}; 

    \addplot+[smooth , mark=triangle, mark repeat=40, color=orange, thick] table [y=y3, x=x, col sep=comma]
    {figs/training_cifar_varying_mal_num.csv};

    \addplot+[smooth, mark=diamond, mark repeat=40, color=green, thick] table [y=y4, x=x, col sep=comma]
    {figs/training_cifar_varying_mal_num.csv};

    \addplot+[smooth, mark=square, mark repeat=40, thick, color=red] table [y=y5, x=x, col sep=comma]
    {figs/training_cifar_varying_mal_num.csv};

    \end{groupplot}
    \node at (7, -1.5){\ref{grouplegend1}};
  \end{tikzpicture}
  \caption{Impact of \attack~on accuracy with varying $\epsilon$}
  \label{fig:varying_mal_num}
\end{figure*}
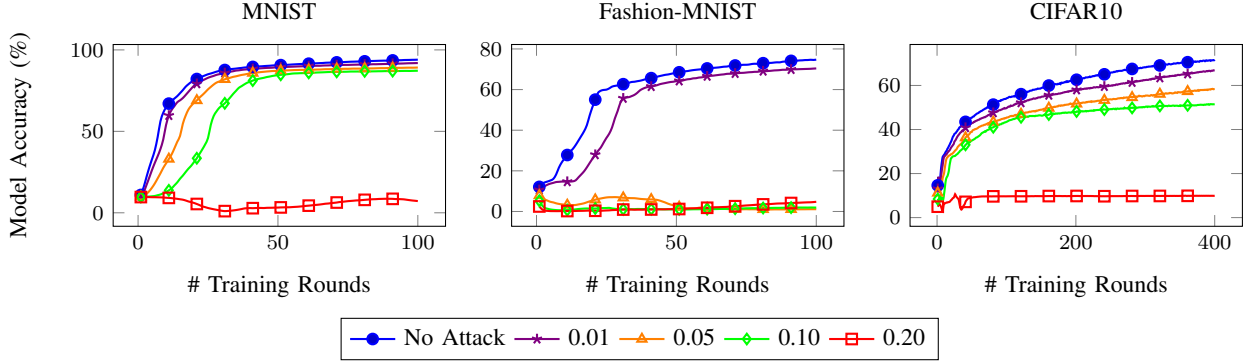
\paragraph{Partial-knowledge setting.} 
We consider only the gradients of the first $\epsilon \cdot n$ samples chosen randomly in each SGD round are seen and controlled by the adversary.
\attack~works effectively in this setting as well as shown in Fig.~\ref{fig:partial_filtering}.  \attack~reduces the model accuracy by $\mathbf{86\%}$, $\mathbf{65\%}$, and $\mathbf{26\%}$ in MNIST, Fashion-MNIST, and CIFAR10 datasets respectively against FILTERING. The results are similar against NO-REGRET (see Fig~\ref{fig:partial_noregret}). The impact of \attack~in the partial-knowledge setting is much more pronounced than that of prior attacks even with full-knowledge.

\begin{mdframed}[backgroundcolor=mygray]
    \attack~ with partial-knowledge is more effective even when compared to prior attacks with full-knowledge.
\end{mdframed}

One expects the attack efficacy to reduce in the partial-knowledge setting compared to the full-knowledge. This is because the adversary has to estimate the direction of the benign mean using only points it can see, i.e., its own vectors. Our result shows that the difference between the two settings, though visible, is relatively small on real datasets.

\paragraph{\attack~performance with varying $\epsilon$.} We have reported all the results at $\epsilon=0.2$ so far. We show the impact of \attack~for varying $\epsilon=\{0.01, 0.05, 0.1, 0.2\}$, the fraction of corrupted samples, in the full-knowledge setting.


Fig.~\ref{fig:varying_mal_num} shows the results. \attack~continues to exhibit a large drop in model accuracy even at lower fractions. For instance, at $\epsilon=0.05$, \attack~lowers the performance by $\mathbf{74\%}$ on Fashion-MNIST. At $\epsilon=0.1$, the drop induced is over $20\%$ for CIFAR10.
The sharpest drops occur at $0.1 \leq \epsilon < 0.2$ for MNIST and CIFAR10.

\begin{mdframed}[backgroundcolor=mygray]
    \attack~can induce over $70\%$ loss of model accuracy even at smaller $\epsilon=0.05$ on some datasets. The sharpest drop in accuracy is between $\epsilon$ of $0.1$ and $0.2$ for others.
\end{mdframed}



\section{Related Work}
\label{sec:related}


Byzantine Robust aggregation has its roots in robust statistics, a subject with a rich history~\cite{huber2004robust,maronna2019robust}. The threat of poisoning attacks on large neural nets has brought urgent attention to it, and in particular, to the challenge of minimizing bias when aggregating high-dimensional vectors. 

\paragraph{Poisoning attacks.}
Poisoning attacks are a long-standing issue for ML in adversarial environments~\cite{newsome2006paragraph}. Their threat was highlighted early in federated learning systems for general ML classifiers~\cite{wang2014man} and neural networks~\cite{shen2016auror}. Since then there has been an evolving cat-and-mouse game between poisoning attacks~\cite{fang2020local, xie2019dba, shejwalkar2021manipulating, bagdasaryan2020backdoor, wang2020attack} and defenses to mitigate them~\cite{nguyen2022flame, sun2019can, wu2020mitigating, xie2022zenops, jia2021intrinsic, jia2022certified, levine2020deep, rosenfeld2020certified, weber2023rab, krauss2023mesas, xie2023unraveling}.
Targeted poisoning attacks aim to train models that misclassify a certain class of inputs~\cite{bhagoji2019analyzing, tolpegin2020data}. Backdoor attacks aim to train models that misclassify inputs that have planted trigger patterns~\cite{bagdasaryan2020backdoor, xie2019dba, wang2020attack}. These attacks also bias the aggregate gradient of the model, however, they do not aim to create optimal bias. Therefore, existing strong \defences~have been shown to mitigate these attacks~\cite{zhu2023byzantine}. 
Untargeted poisoning attacks,  our motivating application,
also create bias but differ in their goal, i.e., to destroy model performance~\cite{aisec2017}.
The unifying vulnerability across poisoning attacks is that they bias gradient vectors used in averaging, with the malicious ones having orientation and magnitude different from the benign vectors~\cite{hong2020effectiveness}. Knowing something about the distribution of benign vectors, therefore, is a universal basis to build defenses.
%
The intuition that outlier removal based on the distributional properties of benign vectors can increase the robustness of modern ML models is evident in early work~\cite{shen2016auror}, but many such defenses lack theoretical analysis and evaluation on non-adaptive attacks. The main issue is how to provably defeat an adaptive attacker, one which carefully adjusts the attack strategy with knowledge of the defense. 



\paragraph{Weak \defences.} 
Many aggregators are accompanied by analyses under adaptive worst-case inputs. They have weak bounds that are proportional to $O(\sqrt{\epsilon \cdot d})$.
Yin et al. proposed coordinate-wise trimmed mean and median for byzantine robust federated learning~\cite{yin2018byzantine}. Concurrently, defenses based on Euclidean distances between vectors, Krum, and geometric median, were proposed~\cite{blanchard2017machine, pillutla2022robust}
and subsequently extended~\cite{guerraoui2018hidden}.
In Krum, for each vector, the sum of distances of the closest $n-n\epsilon -1$ vectors is computed, and the vector with the smallest such aggregate distance is chosen as the mean. Analyses for all of these aggregators give a bias upper bound of $O(\sqrt{\epsilon\cdot d})$~\cite{lai2016agnostic, lugosi2021robust}.

\paragraph{Poisioning attacks against weak \defences.} Untargeted poisoning attacks, especially in federated learning, are the prime nemesis for~\defences.
Bit flipping and label flipping were early attacks that proved to be effective without any defense in place~\cite{yin2018byzantine, paudice2019label}. Several weak~\defences~such as coordinate-wise, trimmed mean, and Krum have demonstrably mitigated these attacks~\cite{yin2018byzantine, blanchard2017machine}. Xie et al.~\cite{xie2020fall} and Fang et al.~\cite{fang2020local} have proposed several attacks that mitigate known weak \defences~as well. They follow a similar strategy of placing corrupted gradients to ours along the opposite direction of the benign gradient and tuning the magnitude of the corrupted gradients to defeat the specific attacks considered. Their bias introduced, however, is far from optimal. Shejwalkar et al. have recently shown that these attacks do not affect the training process much at all when the corruption fraction is low ($\epsilon<0.01$)~\cite{shejwalkar2022back}. Nevertheless, all of these attacks have been recently shown to be mitigated by much stronger defenses, even at higher $\epsilon\in [0.1-0.5)$. This is why our focus has been on strong \defences.

\paragraph{Strong \defences.} These aggregators achieve an $\tilde{O}(\sqrt{\epsilon})$ upper bound on the bias, thus, removing the dependence on $d$~\cite{zhu2023byzantine, zhu2020deconstructing, hopkins2020robust} (see Section~\ref{sec:strong_defences} for details). Zhu et al. show that these aggregators when implemented in practice mitigate all aforementioned weak \defences. We design \attack~to break the guarantees offered strong defenses~\cite{zhu2023byzantine, zhu2020deconstructing, hopkins2020robust}. We give optimal attacks in the low dimensional regime and identify a common computational bottleneck in them. We are not aware of prior work that carefully explains why the bottleneck is somewhat fundamental and identifies its existence as an exploit opportunity. The state-of-the-art practical realization of strong \defence~ is given by~\cite{zhu2023byzantine}. Our attack creates bias proportional to $\sqrt{d}$ in these practical realizations. Our results leave the gap wide open between the ideally desired robustness and that which is practical presently in high dimensional settings.



\paragraph{Robust mean estimation.}
Robust mean estimation in high dimensions lends a theoretical underpinning to byzantine robust federated learning, as explained in Section~\ref{sec:backandprob}. Apart from the \defences~we discussed in Section~\ref{sec:strong_defences}, Zhu et al. propose a robust aggregator using Generative Adversarial Networks that are trained to remove the outliers. This method relies heavily on training and tuning the hyperparameters of GANs to provide to optimal bias guarantees which is not feasible in practical scenarios~\cite{zhu2023byzantine}. Cheng et al. proposed a robust mean aggregator based on Filtering with optimal bias guarantees with time complexity $O(\frac{nd}{\epsilon^6})$ ~\cite{cheng2019high}. This aggregator relies on solving a dual SDP problem with $2$-degree constraints in linear time. However, we are not aware of any such efficient SDP solvers for high $d$. Another work proposes a robust aggregator with optimal bias guarantees with $O(nd\cdot poly(\log d))$ ~\cite{dong2019quantum}. The analysis of this aggregator relies on having auxiliary information about the number of corrupted directions.




\section{Conclusion \& Future Work}


We have shown nearly optimal attacks against practical realizations of strong \defences.  In our experiments, untargeted poisoning attacks using our approach against these algorithms almost completely destroy model performance where previous attacks fail to have much impact. Strong aggregators are thus practically much more vulnerable than anticipated by prior theoretical analysis when working with high-dimensional vectors.
%

We have argued that the vulnerability is fundamental to strong aggregators that are deterministic and aim to work for all feasible distributions generically. Future work can explore provable algorithms to directly improve the computational bottlenecks we highlight, consider randomized defenses, or specialize for features arising in certain gradient distributions. \attack~is not meant to be a stealthy attack; a specialized defense targeting \attack~can detect its signature. Devising stealthier attacks that next adaptation of defenses cannot detect would be interesting. Extending \attack~to other poisoning attacks is another possibility. 

\section*{Acknowledgements}

We are thankful to the anonymous reviewers and our shepherd on the program committee. We also wish to thank Ankit Pensia, Kareem Shehata, and Jason Zhijingcheng Yu for their helpful feedback on previous drafts. 
This research is supported by the research funds of the Crystal Centre at National University of Singapore and the Ministry of Education Singapore grants: Tier-2
grant MOE-T2EP20220-0014 and Tier-1 grant T1 251RES2023. All opinions expressed in the work are those of the authors.

\bibliographystyle{IEEEtran}
\bibliography{IEEEabrv, paper}



\appendices
\section{Theoretical Analysis: Complete Proofs}
\label{sec:full_proofs}

Here we provide the complete proofs for Lemmas in Sections~\ref{sec:analysis} and~\ref{sec:tradeoff}.

\lemla*
\begin{proof}

We are given $X: \{x_1, x_2, \ldots, x_n\}$ and $Y: \{y_1, y_2, \ldots, y_n\}$ and a direction $\hat{s}$. Without loss of generality, let's assume that we corrupt the first $n \cdot \epsilon$ update vectors in the following manner:
\begin{align}
    y_i = -z\hat{s}\text{, } i \in [1,\ldots,\epsilon n]
\end{align}
We have to find this value of $z$.

\begin{align}
    bias &= \frac{1}{n}||\sum x_i - \sum y_i||_2\\
         &\geq \frac{1}{n}\left[\sum\langle x_i, \hat{s}\rangle - \sum \langle y_i, \hat{s} \rangle\right]\\
         &\geq \mu^{\hat{s}} - \mu_c^{\hat{s}}\\
    \mu_c^{\hat{s}} &= \frac{1}{n} \left(\sum_{i=1}^{n\epsilon}\langle y_i, \hat{s}\rangle + \sum_{i=n\epsilon+1}^{n}\langle x_i, \hat{s}\rangle\right)\\
    &= \epsilon(-z) + \frac{1}{n}\sum_{i=n\epsilon + 1}^{n} \langle x_i, \hat{s}\rangle\\
    &= \epsilon (-z) + (1-\epsilon)\left(\mu^{\hat{s}} \pm \frac{\delta \sigma_{\hat{s}}}{\sqrt{n}}\right)\\
    \text{where } &\frac{\delta \sigma_{\hat{s}}}{\sqrt{n}} \text{ is the error with probability } 1 - e^{-\delta^2} \nonumber\\
    \text{using chernoff } &\text{bound, we will add this to $z$ at the end} \nonumber\\
    \implies \mu^{\hat{s}} - \mu_{c}^{\hat{s}} &= \epsilon(z+\mu^{\hat{s}})  \label{eq:t0}
\end{align}

We find that value of $z$ for which the variance along $\hat{s}$ goes from $\sigma^2_{\hat{s}}$ to $\threshold$ and does not exceed it. 

\begin{align}
    \frac{1}{n}\left(\sum_{i=1}^{n\epsilon} \left(-z - \mu_c^{\hat{s}}\right)^2 + \sum_{i=n\epsilon+1}^{n} \left(x_i^{\hat{s}} - \mu_{c}^{\hat{s}}\right)^2 \right) \leq \threshold\\
    \frac{1}{n}\left(\sum_{i=1}^{n\epsilon} \left(-z - \mu_c^{\hat{s}}\right)^2 + \sum_{i=1}^{n}\left(x_i^{\hat{s}} - \mu_c^{\hat{s}}\right)^2 \right. \notag \\ 
    \left. - \sum_{i=1}^{n\epsilon} \left(x_i^{\hat{s}} - \mu_{c}^{\hat{s}}\right)^2 \right) \leq \threshold \label{eq:t1}
\end{align}

We solve for the Equation~\ref{eq:t2} since it implies~\ref{eq:t1} 

\begin{align}
    \frac{1}{n}\left(\sum_{i=1}^{n\epsilon}\left(-z - \mu_c^{\hat{s}}\right)^2 + \sum_{i=1}^{n}\left(x_i^{\hat{s}} - \mu_c^{\hat{s}}\right)^2\right) = \threshold \label{eq:t2}\\
    \frac{1}{n}\left(n\epsilon(-z-\mu_c^{\hat{s}})^2 + \sum_{i=1}^{n}\left(x_i^{\hat{s}}-\mu^{\hat{s}} + \mu^{\hat{s}} - \mu_c^{\hat{s}}\right)^2\right)\\
    \frac{1}{n}\left(n\epsilon\left(-z-\mu_c^{\hat{s}}\right)^2 + \Sigma_{i=1}^{n}\left(x_i^{\hat{s}} - \mu^{\hat{s}}\right)^2 + n\left(\mu^{\hat{s}} - \mu_c^{\hat{s}}\right)^2\right)\\
    \implies \epsilon\left(-z - \mu_c^{\hat{s}}\right)^2 + \sigma_{\hat{s}}^2 + \epsilon^2\left(z + \mu^{\hat{s}}\right)^2 = \threshold \\
    \text{after substituting from }\ref{eq:t0}\notag\\
    \epsilon\left(-z + \epsilon\left(z+\mu^{\hat{s}}\right) - \mu^{\hat{s}} + \delta \right)^2 + \sigma_{\hat{s}}^2 + \epsilon^2\left(z+ \mu^{\hat{s}}\right)^2 = \threshold\\
    \left(z + \mu^{\hat{s}}\right)^2\cdot\left(\epsilon^2 + \epsilon(1-\epsilon)^2\right) = \threshold - \sigma_{\hat{s}}^2 \\
    z + \mu^{\hat{s}} = \sqrt{\frac{\threshold - \sigma_{\hat{s}}^2}{\epsilon^2 + \epsilon(1-\epsilon)^2}}
\end{align}
Therefore, 
\begin{align}
    z = \sqrt{\frac{\threshold - \sigma_{\hat{s}}^2}{\epsilon^2 + \epsilon(1-\epsilon)^2}} - \mu^{\hat{s}} \pm \frac{\delta \sigma_{\hat{s}}}{ \sqrt{n}}\label{eq:z1}
\end{align}

Substituting the value of $z+\mu^{\hat{s}}$ in Equation~\ref{eq:t0}, we can bound the bias as follows.

\begin{align}
   bias &\geq \sqrt{\epsilon}\cdot\sqrt{\frac{\threshold - \sigma_{\hat{s}}^2}{\epsilon + (1-\epsilon)^2}}\\
   bias &\geq \Omega(\sqrt{\epsilon})\cdot\sqrt{\threshold} \text{, }\{\because \threshold > 9 ||\Sigma||_2, \sigma_{\hat{s}} < ||\Sigma||_2^{\frac{1}{2}}\}
\end{align}
\end{proof}

\lemlb*
\begin{proof}
    
    Consider a random direction $\hat{t}$ and say the corruptions are $-z'\hat{t}$ (notice that $z'$ and $\hat{t}$ are different from $z$ and $\hat{s}$). Let's also consider $\threshold$ as the final variance $\bar{\sigma}^2_{\hat{t}}$ that will be attained along the direction $\hat{t}$. Then as long as $\bar{\sigma}^2_{\hat{t}}>\sigma^2_{max}$, the equation $19$ from Lemma~\ref{lem:l1} can be rewritten as 
    \begin{align}
        \bar{\sigma}_{\hat{t}}^2 = K\cdot(\mu^{\hat{t}} + z')^2 + \sigma_{\hat{t}}^2 \text{, } \{K = \sqrt{\epsilon^2 + \epsilon(1-\epsilon)^2}\}
    \end{align}

    Now, say $\hat{s}$ is the direction of benign mean $\mu$ and we place the corruptions along $-\hat{s}$ rather than $-\hat{t}$. Therefore, 

    \begin{align}
        \bar{\sigma}_{\hat{s}}^2 = K\cdot(\mu^{\hat{s}} + z)^2 + \sigma_{\hat{s}}^2 \text{, } \{K = \sqrt{\epsilon^2 + \epsilon(1-\epsilon)^2}\}
    \end{align}

    In this scenario, the corruptions will have a projection along $-\hat{t}$ which with be $z'=\langle-z\hat{s}, -\hat{t}\rangle$ and $\mu$ will have a component $\mu^{\hat{t}}=\langle\mu, \hat{t}\rangle$ along $\hat{t}$. So, if we add corruptions along $\hat{s}$ then the final variance in the direction $\hat{t}$ can be computed by substituting these projections $z'$ and $\mu^{\hat{t}}$ in equation $22$. Hence we get,
    
    \begin{align*}
    \bar{\sigma}_{\hat{t}}^2 &= K\cdot(\langle\mu, \hat{t}\rangle + \langle-z\hat{s}, -\hat{t}\rangle)^2 + \sigma_{\hat{t}}^2\\
    \bar{\sigma}_{\hat{s}}^2 &= K\cdot(\mu^{\hat{s}} + z)^2 + \sigma_{\hat{s}}^2
    \end{align*}
    Notice that, if we choose $\sigma^2_{\hat{s}}=\sigma^2_{max}$ then,
    \begin{align*}
    \sigma_{\hat{t}}^2 &\leq \sigma^2_{\hat{s}}\text{, }\because \sigma^2_{\hat{s}}=\sigma^2_{max}\\
    (\langle -z\hat{s},\hat{-t}\rangle + \langle\mu, \hat{t}\rangle)^2 &\leq (\mu^{\hat{s}}+z)^2\\
    \implies \bar{\sigma}_{\hat{s}}^2 \geq \bar{\sigma}_{\hat{t}}^2\text{, } \forall \hat{t}\\
    \implies ||\Sigma_Y||_2  = \sigma_{\hat{s}}^2 \leq \threshold
    \end{align*}
\end{proof}

\lemlc*
\begin{proof}
    For each chunk $i$, 
    $$
    bias_i \geq k\sqrt{\epsilon} \cdot \sqrt{\threshold}
    $$
    Therefore, 
    \begin{align*}
        bias &\geq \sqrt{\sum_{i=1}^{c}(bias_i)}\\
        bias &\geq k\cdot\sqrt{\epsilon c}\cdot \sqrt{\threshold}
    \end{align*}
\end{proof}

\lemfa*
\begin{proof}For set $Y = \{y_1, \dots, y_{n - n \epsilon}\}$ and $Y' = \{y_1, \dots, y_{n - n\epsilon}, y_{n - n\epsilon}, \dots, y_n\}$. Consider, 
\begin{align}
    \mu &= \frac{1}{n - n\epsilon} \sum_{i=1}^{n - n\epsilon}y_i = \hat{\mu} \\
    \mu' &= \frac{1}{n} \sum_{i=1}^{n} y_i = (1 - \epsilon)\hat{\mu} + \frac{1}{n}\sum_{i=n-n\epsilon+1}^{n} y_i \\
    \mu' &= (1 - \epsilon)\hat{\mu} + \mu_{c} \hspace{0.5cm} \text{where } \mu_{c} = \frac{1}{n}\sum_{i=n-n\epsilon+1}^{n} y_i \label{eq:corrupted_mean}
\end{align}
Each $y_i$ for $i \in [n - n\epsilon + 1, n]$ aligns along one of the directions in $B = \{b_1, b_2\}$, so

\begin{align}
    \mu_{c} &= \frac{1}{n}\left( \sum_{j=1}^{\frac{n \epsilon}{2}}(\sqrt{d} + l_{1j})\cdot b_1 + \sum_{j=1}^{\frac{n \epsilon}{2}}(\sqrt{d} + l_{2j}) \cdot b_2 \right) \\
    \mu_{c} &= \frac{\epsilon \sqrt{d}}{2}(b_1 + b_{2})
\end{align}
Given the set $B$ is chosen randomly without any correlation to $Y$, without loss of generality, we can consider $B$ as unit vectors along any $2$ axes in $\mathbb{R}^d$. For simplicity, we consider $B$ as the unit vectors along the first $2$ axes. Following this, $\mu_{c} \in \mathbb{R}^d$ has first $2$ components as $\frac{\epsilon \sqrt{d}}{2}$ and rest as $0$. 
\begin{align}
    \mu_c = \Big[ \frac{\epsilon \sqrt{d}}{2}, \frac{\epsilon \sqrt{d}}{2},0 , \dots, 0\Big] \label{eq:avg_corruption_vector}
\end{align}

To find the direction of the maximum variance of $Y'$, first, we compute the covariance matrix of $Y'$ as follows.
\begin{align}
    \Sigma_{Y'} &= \frac{1}{n} \sum_{i=1}^{n} (y_i - \mu')^T \cdot (y_i - \mu') 
\end{align}
Substitute $\mu'$ from eq. ~\ref{eq:corrupted_mean}
\begin{align}
    \Sigma_{Y'} &= \frac{1}{n} \sum_{i=1}^{n} (y_i - (1-\epsilon)\hat{\mu} - \mu_c)^T \cdot (y_i - (1-\epsilon)\hat{\mu} - \mu_c) \label{eq:corrupted_covariance_first}
\end{align}
To simplify the eq. ~\ref{eq:corrupted_covariance_first} , consider the following equations. 
\begin{align}
    \text{Cov}(Y) &= \frac{1}{n - n\epsilon}\sum_{i=1}^{n - n\epsilon}(y_i - \hat{\mu})^{T} \cdot (y_i - \hat{\mu}) = \hat{I} \\
    \frac{1}{n - n\epsilon} &\left( \sum_{i=1}^{n - n\epsilon}y_i^T \cdot y_i \right) - \hat{\mu}^T \cdot \hat{\mu} = \hat{I} \\
    \sum_{i = 1}^{n - n\epsilon} y_i^T &\cdot y_i = (n - n\epsilon) ( \hat{I} + \hat{\mu}^T\cdot\hat{\mu}) \label{eq:subs_1}
\end{align}
Further from eq. ~\ref{eq:corrupted_mean}, we can write
\begin{align}
    \sum_{i = 1}^{n}y_i = (n - n \epsilon)\hat{\mu} + n\mu_c \label{eq:subs_2}
\end{align}
Now, expand the RHS of eq. ~\ref{eq:corrupted_covariance_first} by multiplying each term and replacing the following terms using eq. ~\ref{eq:subs_1} and ~\ref{eq:subs_2}
\begin{align}
    \sum_{i=1}^{n} y_i^T\cdot y_i &\rightarrow (n - n\epsilon) (\hat{I} + \hat{\mu}^T\cdot\hat{\mu}) + \sum_{i=n-n\epsilon + 1}^{n} y_i^T\cdot y_i \nonumber\\
    \sum_{i = 1}^{n}y_i &\rightarrow (n - n \epsilon)\hat{\mu} + n\mu_c \nonumber
\end{align}
The simplified expression of $\Sigma_{Y'}$ would be:
\begin{align}
    \Sigma_{Y'} &= \frac{1}{n}\sum_{i=n - n\epsilon + 1}^{n} y_i^T\cdot y_i  - \mu_c^T\mu_c \nonumber\\
    + (1 - \epsilon)\hat{I}  &- (1 - \epsilon)(\mu_c^T \cdot \hat{\mu} +  \hat{\mu}^T \cdot {\mu_c}) + (1-\epsilon)(\epsilon)\hat{\mu}^T \cdot \hat{\mu} \label{eq:covariance_long}
\end{align}

Notice the last three terms of the above expression contain the errors due to sampling from the spherical Gaussian in $\hat{\mu}$ and $\hat{I}$. Next, we write the matrix expression of all terms and bound the variance incurred by the error terms by a constant which is negligible compared to the total. Hence, we show that the unit vector direction that maximizes the variance of $\Sigma_{Y'}$ including the error terms will be approximately close to the exact maximum variance direction computed without them. 

As all the corrupted vectors are along one of the first $2$ axes in $\mathbb{R}^d$, $\sum_{i=n - n\epsilon + 1}^{n} y_i^T\cdot y_i$ is a diagonal matrix with the first $2$ as non-zero and equal to $\sum_{j=1}^{\frac{n \epsilon}{2}} (\sqrt{d} + l_{ij})^2$ for $i \in [1, 2]$. Therefore, the first term of eq. ~\ref{eq:covariance_long} can be written as following where $\gamma = \sum_{j=1}^{\frac{n \epsilon}{2}} l_{ij}^2, \text{ for} i \in [1, 2]$:
\begin{align}
    \frac{1}{n}\sum_{i=n - n\epsilon + 1}^{n} y_i^T \cdot y_i &= \left(\frac{\epsilon d}{2} + \frac{\gamma}{n}\right)
    \begin{bmatrix}
        1 & 0 & 0 & \dots & 0 \\
        0 & 1 & 0 & \dots & 0\\
        0 & 0 & 0 & \dots & 0\\
        \vdots & \vdots & \vdots & \ddots \\
        0 & 0 & 0 & 0 & 0 \\
    \end{bmatrix} 
\end{align}
Further, from eq. \ref{eq:avg_corruption_vector}, $\mu_c^T \cdot \mu_c$ has first $2 \times 2$ terms as non-zero and equal to $\frac{\epsilon \sqrt{d}}{2} \times \frac{\epsilon \sqrt{d}}{2}$. So, we have:
\begin{align}
    \mu_c^T \cdot \mu_c &= \frac{\epsilon^2 d}{4}
    \begin{bmatrix}
    1& 1 & 0 & \dots & 0\\
    1& 1 & 0 & \dots & 0\\
    0& 0 & 0 & \dots & 0\\
    \vdots & \vdots & \vdots & \ddots \\
    0 & 0 & 0 & 0 & 0 \\
    \end{bmatrix} 
\end{align}
Now, for the sampling error terms in eq.~\ref{eq:covariance_long}, assume the sampling error in the first $2$ component of $\hat{\mu}$ is $\hat{\mu}_1, \hat{\mu}_2$. Then, the matrix form of the sampling error terms is as follows:
\begin{align}
     (1 - \epsilon)&(\mu_c^T \cdot \hat{\mu} +  \hat{\mu}^T \cdot {\mu_c}) \nonumber\\
     &= \frac{(1 -\epsilon)\epsilon \sqrt{d}}{2}\cdot 
    \begin{bmatrix}
    2\hat{\mu}_1& \hat{\mu}_1 + \hat{\mu}_2 & 0 & \dots & 0\\
    \hat{\mu}_1 + \hat{\mu}_2&  2\hat{\mu}_2& 0 & \dots & 0\\
    0& 0 & 0 & \dots & 0\\
    \vdots & \vdots & \vdots & \ddots \\
    0 & 0 & 0 & 0 & 0 \\
    \end{bmatrix} 
\end{align}
As set $Y$ is sampled from a spherical Gaussian, it implies that both $\hat{\mu}_1$ and $\hat{\mu}_2$ is less than $O(\frac{\delta}{\sqrt{n}})$ with probability at least $(1 - 2e^{\frac{-\delta^2}{2}})$ using Chernoff bound. Using Frobenius norm as bound for Spectral norm, we can claim the variance incurred by this is O($\frac{\delta}{\sqrt{n}})\cdot\sqrt{d}$. Similarly, the variance incurred by the error term $\hat{\mu}^T\cdot\hat{\mu}$ is $O(\frac{\delta^2}{n})\cdot d$ with probability at least $(1 - (\frac{\delta}{\sqrt{n}})^d)$ using Chernoff bound. And for $\hat{I}$, the spectral norm is bounded by $O(\frac{d}{n})$ ~\cite{vershynin2012close}. 

Now, we combine the exact and sampling terms in eq.~\ref{eq:covariance_long} and write them in matrix form where $\alpha = \frac{\epsilon d}{2}(1 - \frac{\epsilon}{2}) + \frac{\gamma}{n}$, $\beta = -\frac{\epsilon^2 d}{4}$, and $\Sigma_{err}$ is the combined covariance matrix for sampling error terms.
\begin{align}
    \Sigma_{Y'} =     
    \begin{bmatrix}
        \alpha & \beta & 0 & \dots & 0\\
        \beta& \alpha & 0& \dots & 0\\
        0 & 0 & 0 & \dots & 0\\
        \vdots & \vdots & \vdots & \ddots \\
        0 & 0 & 0 & 0 & 0 \\
    \end{bmatrix}  + \Sigma_{err}
\end{align}
Finally, we have a matrix form of the exact covariance term and upper bound of the error terms. Next, take an arbitrary unit vector $v = [v_1, v_2, \dots, v_d]$ and the variance of set $Y'$ along $v$, denoted by $\sigma_v^2$, is as follows: 
\begin{align}
    \sigma_v^2 &= \langle v, \Sigma_{Y'} \cdot v \rangle \\
    \sigma_v^2 &= \alpha(v_1^2 + v_2^2) + 2\beta(v_1v_2) + \langle v, \Sigma_{err} \cdot v \rangle
\end{align}
Notice the exact part of the variance scales with $d$ which is much larger than the variance incurred by error terms. Since $\alpha$ is positive and $\beta$ is negative, the exact part is maximum when $v_1^2 + v_2^2$ gets equal to $1$ and $v_1v_2$ attains its minimum value. Hence, the exact part attains its maximum value at $v_1 = \frac{1}{\sqrt{2}}$, $v_2 = -\frac{1}{\sqrt{2}}$ or vice-versa which is the direction of difference of vectors in set $B$. Now, we argue that the maxima of the whole expression should be close to the above solution. Consider another unit vector $v'\in \mathbb{R}^d$ close to the above solution $v=[\frac{1}{\sqrt{2}}, -\frac{1}{\sqrt{2}}]$ such that their dot product is $\rho$, we can write 
\begin{align}
    v_1' - v_2' = \sqrt{2}\rho\\
    v_1'^2 + v_2'^2 - 2v_1'v_2' = 2\rho^2 
\end{align}
Given $\alpha > |\beta|$, the decrease in variance incurred by exact terms is at least $2(1 - \rho^2)|\beta|$. $v'$ is the maxima for the whole expression as long this decrease is less than the variance incurred by error terms. As shown earlier, the maximum value of $\langle v, \Sigma_{err} \cdot v \rangle$ is $O(\frac{\delta^2}{n}\cdot d)$ with probability at least $1 - 2e^{\frac{-\delta^2}{2}}$. For $v'$ to be the direction of maximum variance, 
\begin{align}
    2(1-\rho^2)|\beta| \leq O\left(\frac{\delta^2}{n}\right)\cdot d\\
    2(1-\rho^2) \frac{\epsilon^2 d}{4} \leq O\left(\frac{\delta^2}{n}\right)\cdot d\\
    \sqrt{1 - \frac{2}{\epsilon^2}O\left(\frac{\delta^2}{n}\right)} \leq \rho
\end{align}
Hence, ignoring $\frac{2}{\epsilon^2}$ we have shown that the direction of the maximum variance of $Y'$ has an angle of less than $cos^{-1}\left(\sqrt{1 - O(\frac{\delta^2}{n}})\right)$ from the difference of vectors in set $B$ with probability at least $1 - 2e^{\frac{-\delta^2}{2}}$.
\end{proof}

\lemfb*
\begin{proof}
Similar to the proof of Lemma ~\ref{lem:1.1}, since the set $B$ is chosen randomly without any correlation to $Y$, we can take $B$ as any $2$ axes in $\mathbb{R}^d$. For simplicity of calculations, let's take $B$ as the first $2$ axes. Then, we can write $\mu'$ can be written as the following using eq.\ref{eq:corrupted_mean} and \ref{eq:avg_corruption_vector},
\begin{align}
    \mu' = [\frac{\epsilon\sqrt{d}}{2}, \frac{\epsilon\sqrt{d}}{2}, 0, \dots, 0] + (1 - \epsilon)\hat{\mu}
\end{align}
Notice for all $y_i \in Y'$, $\langle \frac{y_i}{||y_i||_2}, \hat{\mu}\rangle$ follows a $1$-D Gaussian distribution with mean $0$ and variance $\frac{1}{\sqrt{n}}$. Therefore, it would be $O(\delta)$ with probability of $1 - e^{\frac{-n\delta^2}{2}}$
So, for all corrupted vectors $y_i \in Y'$ with probability $1 - e^{\frac{-n\delta^2}{2}}$,  
\begin{align}
\langle \frac{y_i}{||y_i||_2}, \mu'  \rangle &= \langle b_j, \mu' \rangle \hspace{0.4cm} b_j \in B \\
\langle \frac{y_i}{||y_i||_2}, \mu'  \rangle &= \frac{\epsilon \sqrt{d}}{2} \pm (1 - \epsilon)O(\delta) \label{eq:corrupted_proj}
\end{align}
For uncorrupted $y_i \in Y'$, consider $y_1 = [y_{11}, y_{12}, \dots, y_{1d}]$, 
\begin{align}
    \langle \frac{y_1}{||y_1||_2}, \mu' \rangle = \frac{\epsilon \sqrt{d}}{2||y_1||_2} (y_{11} + y_{12}) + (1 - \epsilon)O(\delta)\label{eq:20}
\end{align}
For such uncorrupted $y_i$ in $Y'$, $\frac{\sqrt{d}}{||y_i||_2}$ is a constant $t\approx 1$ since each of them is a random sample from $d$-dimensional spherical gaussian with mean as $0$ and covariance as $I$, hence each of them has an $L_2$ norm of $\sqrt{d}$ on expectation.
\begin{align}
    \langle \frac{y_1}{||y_1||_2}, \mu'\rangle = \frac{\epsilon t}{2}(y_{11} + y_{12}) + (1 - \epsilon)O(\delta)\label{eq:21}
\end{align} 
For the projection in eq. ~\ref{eq:21} to be more than $\frac{\epsilon\sqrt{d}}{2}$, at least one of the two components should be greater than $\frac{\sqrt{d}}{2 t} - O(\delta)$. But all components of an uncorrupted vector in $Y'$, as $y_1$, follow a single dimensional Gaussian distribution with mean as $0$ and variance as $1$. Then, using Chernoff bounds for $1$-D Gaussian, we can calculate the probability of any of the first $2$ components being more than $\frac{\sqrt{d}}{2 t} - O(\delta)$ as follows,
\begin{align}
    \text{Pr}(y_{11} \text{ or } y_{12} \geq \frac{\sqrt{d}}{2 t} - O\left(\delta\right)) \leq 2 \cdot e ^{(-\frac{d}{8t^2} + O(\delta^2))}
\end{align}
Note that the probability of the event mentioned above is exponentially small in $d$. Hence, we can claim that the projection of $\mu'$ along corrupted vectors in $Y$ is larger compared to the uncorrupted vectors with extremely high probability. Furthermore, given $||f(Y') - \mu||_2 \leq \tau \sqrt{\epsilon}$ for some constant $\tau$, for $\mu = 0$ we can write, 
\begin{align}
\langle \tilde{\mu}, v\rangle &\leq \tau \sqrt{\epsilon} \hspace{1cm} \forall v \in \mathbb{R}^d, ||v||_2 = 1 \label{eq:23}
\end{align}
\
From eq. \ref{eq:corrupted_proj} and eq. \ref{eq:23}, for all corrupted vectors in $Y'$
\begin{align}
\langle \mu' - \tilde{\mu}, \frac{y_i}{||y_i||_2} \rangle &\geq \left(\frac{\epsilon}{2}\sqrt{d} - \tau\sqrt{\epsilon}  + (1 - \epsilon) O(\delta)\right) \label{eq:24}
\end{align}
\
And for all other uncorrupted vectors in $Y'$ with probability $1 - n e^{ \frac{-n\delta^2}{2}}$ (ignoring negligibly small terms than $n e^{ \frac{-n\delta^2}{2}}$), 
\begin{align}
    \langle \mu' - \tilde{\mu}, \frac{y_i}{||y_i||_2} \rangle < \left(\frac{\epsilon}{2}\sqrt{d} - \tau\sqrt{\epsilon}  + (1 - \epsilon) O(\delta)\right) \label{eq:25}
\end{align}
It implies that if we arrange all the vectors of $Y'$ in the order of projection of $\mu' - \tilde{\mu}$ along them, then top $n \cdot \epsilon$ vectors will be the set of corrupted vectors $C$ with very high probability. 
Given the corrupted set contains $\frac{n \epsilon}{2}$ along each vector in $B$. It implies that for each corrupted vector there will be $\frac{n \epsilon}{2}$ other corrupted vectors perpendicular to it. For any such mutually perpendicular corrupted vectors, one of them is along $b_1$ and the other is along $b_2$. Then, 
\begin{align}
    \frac{c_i}{||c_i||_2} - \frac{c_j}{||c_j||_2} = \pm (b_1 - b_2)  \hspace{0.5cm} c_i, c_j \in C, \langle c_i, c_j \rangle = 0
\end{align}
This confirms that $\frac{c_i}{||c_i||_2} - \frac{c_j}{||c_j||_2} $ equals to the difference of vectors in set $B$ and hence aligns with the direction of maximum variance of $Y'$ as well (Lemma ~\ref{lem:1.1}). The reduction algorithm ~\ref{alg:reduction_algo} creates the set $C$ and finds two perpendicular vectors in $C$. Finally, the algorithm returns the difference of these vectors normalized by their norm which is equal to the difference of vectors in $B$. 
\end{proof}

\section{Addressing Meta-Review Comments}
\label{sec: dnceval}
We address the noteworthy concern 2) raised in the meta-review by explaining why we did not evaluate DnC defense (Alg. 2 in \cite{shejwalkar2021manipulating}). This is because DnC does not satisfy the definitions of a strong robust aggregator, to the best of our knowledge. 

We point readers to our meta Algorithm~\ref{alg:existing_algos} outlined a framework for strong robust aggregators, featuring a loop (Line 3) for the iterative removal of corrupted vectors.  DnC does not have such an iterative step. Note that prior analyses of strong robust aggregators underscore that iterations at least equal to the number of corrupted vectors are necessary to provide provable guarantees in the worst-case scenario (i.e. where all corruptions are mutually orthogonal).

To illustrate, we give a straightforward attack against DnC. The absence of the iterative steps means DnC can only eliminate corrupted vectors aligned with the direction of the maximum eigenvector. It fails to detect corrupted vectors aligned with other eigenvectors, such as the eigenvector with the second-highest variance. This limitation arises because DnC never recalculates eigenvectors after removing corrupted vectors along the eigenvector direction with the largest variance, unlike other existing strong aggregators.

Building on this insight, our attack works as follows: Begin by choosing vector, $b_1$, of dimension $d$ (the dimension of benign gradients) with binary values $\{0,1\}$ chosen randomly, and its complement vector, $b_2$. This strategy ensures that the dot product of $b_1$ and $b_2$ is zero, as well as the dot product of any subset of dimensions of $b_1$ and $b_2$. Next, place a single corrupted vector in the direction of $b_1$ from the mean, at a distance of $\beta \cdot ||avg||$, where $||avg||$ represent the average distance of benign vectors from the mean. Subsequently, position the remaining corrupted vectors in the direction of $b_2$ from the mean, at a distance of $c \cdot \beta \cdot ||avg||$. Select a small value for $c$ to ensure that the maximum eigenvector after corruption aligns closely with the direction of $b_1$ from the mean. Under this attack, DnC only detects and filters out a single corrupted vector along $b_1$ from the mean, while marking the remaining vectors as inliers and computing the arithmetic mean with them included. The magnitude of the resulting bias incurred by such a corruption strategy depends on the parameter $\beta$ chosen by the attacker and can be scaled arbitrarily. We have experimentally verified this claim by considering a subset of $1000$ dimensions of gradients while training a CNN on CIFAR10 with $\epsilon=0.2$, as described in section \ref{sec:eval_setup}. We corrupted this set using the aforementioned attack strategy with $c=0.02$ and employed DnC to compute its robust mean at one step. Fig. \ref{fig:dnc_attack} illustrates the bias incurred by the attack on the same set with varying $\beta$. It scales up with $\beta$ and can be further scaled arbitrarily. Furthermore, we confirmed that all vectors along $b_2$ from the mean are marked as inliers by DnC for all the varying values of $\beta$. The code for our attack and our implementation\footnote{The artefacts of the original paper did not have the algorithm's implementation. Hence, we re-implemented it and privately corresponded with one of the authors of that work to check its correctness.} of DnC is provided in~\cite{hidracode}.

The DnC defense also proposes to sample a subset of dimensions and considers only them to compute the largest eigenvector to filter outlier vectors in one shot. This is fundamentally different from iterative removal of corrupted vectors by computing a maximum eigenvector in every iteration as done in strong robust aggregators. Our presented attack here will work even if a smaller or larger subset of dimensions are used to compute eigenvectors.


\begin{figure}
    \centering
  \begin{tikzpicture}
    \begin{groupplot}[
group style = {group size = 1 by 1, horizontal sep = 25pt, vertical sep = 35pt, xlabels at=edge bottom, ylabels at=edge bottom},
      width=6cm,
      height=4cm,
    ]

    \nextgroupplot[
      title={},
      xlabel={$\beta$ (parameter in the attack)},
      ylabel={bias ($L_2$ norm)},
      grid style=dashed,
      legend style={column sep = 1pt, legend columns = 2, legend to name = grouplegend1, font=\small},
      tick label style={font=\footnotesize}, 
      label style={font=\small},  
      title style={font=\small},
    ]
    
    \addplot+[smooth, mark=*, thick] table [y=y1, x=x, col sep=comma] {figs/dnc_attack.csv};



    


    

    \end{groupplot}
  \end{tikzpicture}
  \caption{Bias (in $L_2$ norm) incurred by the proposed attack against DnC with varying $\beta$ (parameter in the attack).}
  \label{fig:dnc_attack}
\end{figure}
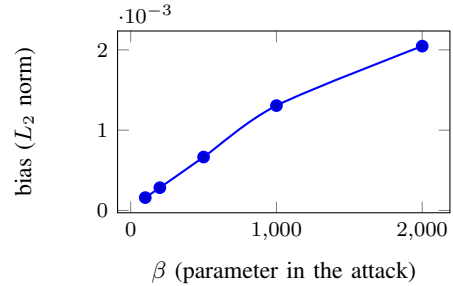

\newpage
\section{Meta-Review}

The following meta-review was prepared by the program committee for the 2024
IEEE Symposium on Security and Privacy (S\&P) as part of the review process as
detailed in the call for papers.

\subsection{Summary}
This paper introduces \attack, a novel attack that targets strong robust aggregators in high-dimensional federated learning environments. Such aggregators provide an upper bound on the bias against poisoning attacks that aim to corrupt a fraction of inputs. However, they introduce a computational bottleneck that limits their application to high-dimensional data like those present in most deep learning settings. To address the bottleneck, practical adaptations split dimensions into disjoint chunks and operate on each chunk individually. \attack~biases the aggregate of each chunk in the dimensionally split robust aggregators, and the paper shows analytically and experimentally that the resulting total bias approaches theoretical upper bounds and degrades model performance.

\subsection{Scientific Contributions}
\begin{itemize}
\item Provides a Valuable Step Forward in an Established Field
\item Identifies an Impactful Vulnerability
\end{itemize}

\subsection{Reasons for Acceptance}
\begin{enumerate}
\item \attack~reveals a new vulnerability in the practical implementation of strong robust aggregators.
\item The paper provides solid theoretical analysis and empirical results that match the theory.
\item The results demonstrate that \attack~achieves near-optimal bias against robust aggregators in untargeted poisoning attacks.
\end{enumerate}

\subsection{Noteworthy Concerns} 
\begin{enumerate} 
\item The paper does not propose countermeasures or mitigation strategies against \attack.
\item The performance of \attack~against DnC-based federated learning, a potential defense, is not evaluated.
\item The evaluation focuses on a single architecture, and the impact that this has on the severity of the attack remains uncertain.
\end{enumerate}

\section{Response to the Meta-Review} 
Thank you reviewers for providing helpful reviews. We have addressed the meta-review comment 2) above regarding DnC in our Appendix~\ref{sec: dnceval}, showing that it does not satisfy the guarantees of a strong aggregator. We leave addressing the remaining comments for future work.

\end{document}